\documentclass[12pt]{article}
\usepackage[T1]{fontenc}
\usepackage{babel}
\usepackage{setspace}
\usepackage{eurosym}
\usepackage{xurl}
\usepackage{amsfonts}
\usepackage{nicefrac}
\usepackage{microtype}
\usepackage{hyperref}
\usepackage{fullpage}
\usepackage{graphicx}
\usepackage{booktabs}
\usepackage[round,authoryear]{natbib}
\usepackage{mathrsfs}
\usepackage{amsmath,amssymb}
\usepackage{changepage}
\usepackage{makecell}
\usepackage{microtype}
\usepackage{amsthm}
\usepackage{changepage}
\usepackage{hyperref}
\usepackage{listings}
\usepackage{graphicx}
\usepackage{subcaption}
\usepackage{bbm}
\usepackage{mathtools}
\usepackage{amsthm}
\usepackage{framed}
\usepackage{amsfonts}
\usepackage[title]{appendix}
\usepackage{setspace}
\newcommand{\Eb}{\mathbb{E}}

\newcommand{\RR}{\mathbb{R}}
\newcommand{\toP}{\xrightarrow{p}}
\newcommand{\bern}{\mathrm{Bern}}

\newcommand{\cD}{\mathcal{D}}

\newcommand{\cN}{\mathcal{N}}
\newcommand{\cA}{\mathcal{A}}
\newcommand{\cS}{\mathcal{S}}
\newcommand{\cH}{\mathcal{H}}

\newcommand{\Var}{\mathrm{Var}}
\newcommand{\Cov}{\mathrm{Cov}}

\newcommand{\KL}{{\mathrm{KL}}}

\DeclareMathOperator*{\argmin}{arg\,min}

\usepackage[normalem]{ulem}  %

\newif\ifclminor
\clminortrue

\ifclminor  %
    \newcommand{\clcomment}[1]{{\color{charlesHotPink}{[\textbf{Charles:} #1]}}}

\else  %
    \newcommand{\clcomment}[1]{}

\fi

\usepackage{url}
\usepackage{import}

\usepackage{amsmath, amssymb, bbold, mathtools}
\usepackage{footnote}
\usepackage{acronym}
\usepackage{enumerate}
\usepackage[capitalise]{cleveref}
\usepackage{multirow}
\usepackage{subcaption}
\usepackage{graphicx}

\usepackage{tikz}
\usetikzlibrary{shapes.geometric}
\usetikzlibrary{matrix,positioning}
\usepackage{pgfplots}
\pgfplotsset{compat=1.15}

\usepackage{mathrsfs}
\usetikzlibrary{arrows}
\usepackage{array}

\usepackage[noline,ruled]{algorithm2e}

\usepackage{amsthm}
\theoremstyle{plain}
\newtheorem{theorem}{Theorem}
\newtheorem{lemma}{Lemma}

\newtheorem{proposition}{Proposition}
\theoremstyle{definition}
\newtheorem{definition}{Definition}
\newtheorem{remark}{Remark}
\newtheorem{example}{Example}

\newtheorem{assumption}[theorem]{Assumption}
\crefname{theorem}{Theorem}{Theorems}
\crefname{lemma}{Lemma}{Lemmas}
\crefname{proposition}{Proposition}{Propositions}
\crefname{corollary}{Corollary}{Corollaries}
\crefname{definition}{Definition}{Definitions}

\usepackage{algorithm}
\usepackage{algpseudocode}

\title{A Bayesian Framework for Human–AI Collaboration: Complementarity and Correlation Neglect
}

\newcommand{\eachauthor}[3]{
    \begin{tabular}{@{}c@{}}
        \small \textbf{#1}\\[-0.5em]  %
        \small #2\\[-0.5em]  %
        \small \texttt{#3}  %
    \end{tabular}
}

\author{%
\centering
\setlength{\tabcolsep}{10pt}
\begin{tabular}{ccc}
\eachauthor{Saurabh Amin}{MIT}{amins@mit.edu}
&
\eachauthor{Amine Bennouna}{Northwestern University}{amine.bennouna@northwestern.edu}
&
\eachauthor{Daniel Huttenlocher}{MIT}{huttenlocher@mit.edu}
\\[1.8em]
\eachauthor{Dingwen Kong}{MIT}{dingwenk@mit.edu}
&
\eachauthor{Liang Lyu}{MIT}{lianglyu@mit.edu}
&
\eachauthor{Asuman Ozdaglar}{MIT}{asuman@mit.edu}
\\
\end{tabular}%
\vspace{0.2in}  %
}

\begin{document}
\pagenumbering{gobble} 
\maketitle

\begin{abstract}
We develop a decision-theoretic model of human–AI interaction to study when AI assistance improves or impairs human decision-making. A human decision-maker observes private information and receives a recommendation from an AI system, but may combine these signals imperfectly. We show that the effect of AI assistance decomposes into two main forces: the marginal informational value of the AI beyond what the human already knows, and a behavioral distortion arising from how the human uses the AI’s recommendation. Central to our analysis is a micro-founded measure of informational overlap between human and AI knowledge. We study an empirically relevant form of imperfect decision-making—correlation neglect—whereby humans treat AI recommendations as independent of their own information despite shared evidence. Under this model, we characterize how overlap and AI capabilities shape the Human-AI interaction regime between augmentation, impairment, complementarity, and automation, and draw key insights.
\end{abstract}

\pagenumbering{arabic}

\def\humancolor{gray}
\def\aicolor{blue}
\def\jointcolor{red}

\def\humanshadecolor{gray!15}
\def\aishadecolor{blue!10}
\def\jointshadecolor{red!10}

\def\humanshadedarkcolor{gray!30}
\def\aishadedarkcolor{blue!20}
\def\jointshadedarkcolor{red!20}

\def\AbilityAxisLabel{AI capability ($\tau_A$)}
\def\LossAxisLabel{Expected loss}
\def\OverlapAxisLabel{Overlap coefficient ($\lambda$)}

\def\HumanCurveLabel{Human alone}
\def\AICurveLabel{AI alone}
\def\JointCurveLabel{Human + AI}

\def\WithholdActionLabel{Impairment}
\def\AugmentActionLabel{Complementarity}
\def\AugmentActionLabelBreak{Complemen-\\tarity}
\def\AutomateActionLabel{Automation}

\def\HtoHAlabel{\tau_{\rm aug}(\lambda)}
\def\HAtoAlabel{\tau_{\mathrm{auto}}(\lambda)}
\def\HtoAlabel{\tau_H}

\def\ThinCurveWidth{0.6pt}
\def\MediumCurveWidth{1.0pt}
\def\ThickCurveWidth{2.4pt}

\pgfplotsset{
    intersect/.style={
        only marks, mark=*, mark size=3pt, black,
        forget plot
    },
    axis to point/.style={
        line width=0.6pt,
        densely dotted, 
        black,
        forget plot
    },
    jointshade/.style={
        draw=none, 
        fill=\jointshadecolor,
        forget plot, 
        on layer=axis background
    },
    humanshade/.style={
        draw=none, 
        fill=\humanshadecolor,
        forget plot, 
        on layer=axis background
    },
    aishade/.style={
        draw=none, 
        fill=\aishadecolor,
        forget plot, 
        on layer=axis background
    },
    jointshade dark/.style={
        jointshade,
        fill=\jointshadedarkcolor,
    },
    humanshade dark/.style={
        humanshade, 
        fill=\humanshadedarkcolor,
    },
    aishade dark/.style={
        aishade, 
        fill=\aishadedarkcolor,
    },
}

\section{Introduction}
Artificial intelligence technologies are advancing at a rapid pace, with transformative potential across many domains. AI’s greatest promise lies in its ability to amplify human decision-making: it can sift through vast volumes of data with unprecedented scale and extract task-relevant insights and recommendations. Yet the value of these recommendations ultimately depends on how they are interpreted and used by human decision-makers.

Consider an AI recommendation system used in medical diagnosis. Such a system may base its recommendations on a comprehensive analysis of prior patient outcomes and historical physician decisions. Yet doctors differ in expertise, decision styles, and how they use algorithmic advice. Whether AI improves outcomes, therefore, depends critically on how physicians integrate its recommendations into their own judgment.
For instance, If a physician fails to recognize that an AI recommendation already incorporates key test results and redundantly supplements it, systematic errors may arise.
Conversely, if the AI omits important patient-specific information and the physician over-trusts its output while discounting their own assessment, diagnostic errors may again occur.
This illustrate that the effectiveness of AI-assistance hinges not only on model accuracy, but also on the structure of human–AI interaction.

Prior literature, as we discuss in the next section, documents cases in which AI assistance improves human decision-making, as well as cases in which it degrades performance. In some settings, full automation—where AI acts alone without a human in the loop—has been shown to outperform human–AI collaboration altogether. A natural question is then when and why each of these regimes arises.

This paper adopts a decision-theoretic framework to study how AI assistance affects human decision-making. Central to our analysis is an explicit model of human–AI interaction, which allows us to characterize the informational and behavioral forces that govern distinct interaction regimes: \textit{augmentation}, when AI assistance improves human decision-making; \textit{impairment}, when AI assistance degrades it; \textit{automation}, when AI should replace human decision-making outright; and \textit{complementarity}, when AI-assisted humans outperform either humans or AI alone.

\paragraph{Contributions and Main Insights.}

We introduce a Bayesian framework to study human–AI interaction. We consider a setting in which a human decision-maker observes private information about an uncertain state and additionally receives a signal from an AI before taking an action. Critically, the human does not necessarily combine the two signals optimally.

Our first insight is a general decomposition of the realized benefit of AI assistance into two competing forces: the \textit{marginal value of AI information}, which captures the incremental decision value beyond what the human already knows, and a \textit{behavioral penalty}, which captures distortions arising from the human’s imperfect decision rule using the AI signal. We show that AI assistance impairs performance precisely when this behavioral penalty exceeds the marginal informational benefit.

Second, to study how these two forces shape human–AI interaction regimes and when each dominates, we introduce the central notion of \textit{information overlap}.
Specifically, we develop a micro-founded model of information overlap between human and AI signals. Both the human and the AI aggregate subsets of a common pool of primitive information sources. This structure yields a natural and interpretable overlap coefficient: the fraction of AI information that is redundant given human expertise. Importantly, when AI acquires information from a fixed pool of accessible sources, we show that informational overlap is approximately invariant to AI capability growth, and is thus a stable, \textit{domain-specific} parameter. This allows us to independently vary AI capability and informational overlap when analyzing joint performance. 

Third, we instantiate the human’s imperfect decision rule with an empirically relevant form: correlation neglect \cite{Agarwal2023NBERradiology,Enke2019CorrNeglect}, whereby humans treat AI recommendations as conditionally independent of their own information even when both draw on shared evidence. We explicitly characterize how information overlap and correlation neglect jointly determine when decisions are best made by the human alone, by the AI alone, or through the AI-assisted human.

Depending on the degree of overlap, three qualitatively distinct phase-transition patterns emerge. (i) When overlap is low, even imperfect humans benefit from AI access, and the system transitions smoothly from complementarity to automation as AI improves, with automation eventually dominating due to the behavioral penalty faced by AI-assisted humans. (ii) When overlap is intermediate, weak AI can strictly impair human performance due to double-counting of shared information, creating an ``impairment'' region in which withholding AI is optimal until AI capability becomes sufficiently high for complementarity to be achieved. As AI capabilities grow even stronger, automation ultimately prevails. (iii) When overlap is high, the behavioral penalty is so severe that AI assistance is never preferred, and the optimal decision-making system transitions directly from human-only (impairment) to AI-only decision-making (automation). The thresholds separating these regimes admit explicit closed-form expressions in terms of the model primitives.

A striking insight across the different phase-transition patterns is that automation is ultimately unavoidable under improving AI capability and stagnant human use of AI. This highlights that human learning in the effective use of AI is critical for pushing the automation frontier outward.

Finally, our analysis reveals a reversal of the classical informativeness result. While additional information always weakly improves decisions for a Bayesian-rational decision-maker, this monotonicity fails under correlation neglect. In particular, there exist parameter regimes in which revealing an informative AI signal strictly worsens human decision-making. This reversal does not arise because the AI signal lacks informational content, but because humans misperceive the dependence between their own information and the AI’s recommendation.

Taken together, these insights contribute to the debate on whether increasingly powerful AI systems should assist or automate decisions. Rather than viewing this choice as driven solely by AI capability, our results emphasize the central role of informational overlap and human misspecficiation in designing optimal forms of human–AI collaboration.

\section{Related Literature}
A growing body of work studies how to combine AI predictions with human expertise~\cite{Kleinberg2017QuartHumanMachine, Mullainathan2021QuartHealthcare}. Much of the focus of this literature is to study whether the combined human-AI system exhibits {\it complementarity}, performing better than human or AI alone. Even though performance improvements from AI recommendations is recorded in some settings, notably in medical imaging~\cite{Tschandl2020NatureSkinCancer}, image classification~\cite{Steyvers2022PNASbayesian}, and recently for customer assistance in a large-scale deployment in call centers~\cite{Brynjolfsson2025QuartGenAIWork}, several experimental and empirical studies show this goal is hard to achieve. For example,~\cite{Bansal2019AAAIupdates} shows that updates that can increase AI’s prediction performance may hurt team performance by increasing overreliance on AI. Adding explanations to the AI decisions does not appear to reduce overreliance, and in fact some studies showed it may even increase it~\cite{Bansal2021CHIcomplementarity}, motivating use of cognitive forcing interventions to motivate people to engage more thoughtfully with AI recommendations~\cite{Bucinca2021HCIcognitiveForcing}. Recent work~\cite{Agarwal2023NBERradiology} shows empirically that providing AI predictions to radiologists worsens diagnostic quality in certain regimes, precisely because humans must integrate AI information with their own and do so imperfectly.

Recent theoretical work has formalized when human-AI complementarity can be achieved using either models for optimally combining information held by different agents or through “learning-to-defer” approaches. In combination frameworks, the combined system’s loss is expressed as a weighted average of human and AI losses~\cite{Donahue2022FAccTcomplementarity} or a convex combination of judgments~\cite{Rastogi2023HCompComplementarity}, with weights taken as given rather than derived from how a human processes AI information. \cite{Donahue2022FAccTcomplementarity} shows that complementarity can be achieved when there are subsets of information regions where each agent has better information. \cite{Alur2024NeurIPSHumanExpertise} introduces a method to identify information regions that are indistinguishable to an agent, in which case signal from the prediction of the other agent can then be leveraged to improve predictions overall. These approaches require data from the joint distribution of agent signals. \cite{Peng2025AAAINoFreeLunch} showed a negative result in the absence of such joint information that shows, except under restrictive conditions, the impossibility of having a combined prediction better than each agents calibrated predictions.

Learning-to-defer methods take a two-stage algorithmic triaging approach, deciding when an automated model can pass the decision to a human without sharing any of its information \cite{Madras2018NeurIPSLearningToDefer, Mozannar2020ICMLLearnDeferEstimators, Lykouris2025learningdefercongestedsystems}. A key difference of these works from ours is that they assume the human decision function is exogeneous and abstract away from endogenous changes in human beliefs in response to AI information.

A closely related paper is \cite{Agarwal2025NBERsufficientStatistic}, which develops a sufficient-statistic approach to designing optimal human-AI collaboration that bypasses the structural modeling of human behavior. Their approach is powerful but relies on the assumption that the probability a human makes a correct classification depends only on the posterior probability given the AI signal. This assumption enables a reduced-form approach to optimizing over the full space of disclosure policies, but as they note, it is typically violated when human and AI signals are conditionally dependent, which is a key case of interest in our paper. Our framework decomposes the forces that govern the value of human-AI collaboration into informational overlap and behavioral misspecification and determines optimal deployment modes as a function of AI capability and signal dependence.

The behavioral mechanism studied in our paper builds on experimental evidence that people systematically neglect correlation when aggregating information from multiple sources, demonstrated in~\cite{Enke2019CorrNeglect}. This behavioral pattern extends to broader economic settings: \cite{eyster2010correlation} document correlation neglect in financial decision-making,~\cite{levy2015correlation} show it affects voting behavior, paradoxically improving aggregate information aggregation by inducing voters to overweight evidence relative to preferences, and~\cite{levy2022persuasion} show that a persuader can fully exploit correlation neglect to manipulate beliefs.

Our work relates to the information design literature~\cite{Kamenica2011AEABayesianPersuasion, Bergemann2016AEAInfoDesignPersuasion}, though it addresses a distinct problem. In standard information design, a sender chooses a signal structure to influence a Bayesian receiver's actions. In our setting, the AI signal is exogenous; its precision and overlap with human information are primitives of the environment. The question is which decision-maker should be preferred: the AI-assisted human, the human-only, or the AI-only. Our value-of-information analysis connects to the literature on signal complements and substitutes~\cite{borgers2013signals}. For a Bayesian decision-maker, the answer is immediate: 
revealing an informative signal to an aligned decision-maker cannot reduce decision quality. This conclusion breaks down under correlation neglect and our contribution is to analyze its implications for human-interaction. When the human treats correlated signals as independent, 
various regimes can occur, including AI-assistance worsening human decision quality. In our framework, we characterize how informational overlap and AI capability jointly determine the interaction regime. 

Our model assumes a human decision-maker receives an AI signal and must integrate it with their own information about an uncertain state. Two modeling innovations distinguish our approach. First, we parameterize informational overlap through a coefficient measuring the share of AI information that is redundant given human expertise, ranging from zero (the disjoint-information setting of learning-to-defer literature literature) to one (complete redundancy), with intermediate values capturing realistic cases where human and AI draw partially on shared evidence. Second, we incorporate correlation neglect: the human treats partially overlapping signals as conditionally independent, leading to double-counting of shared evidence. These features enable identifying regimes where different modes of collaboration are optimal: when overlap is low, even a correlation-neglect human always benefits from AI access. When overlap is intermediate, weak AI harms the human through double-counting, generating a region where withholding AI is optimal; only sufficiently capable AI overcomes this penalty, enabling complementarity. This withholding region arises endogenously from the interaction of information overlap and human misspecification, a mechanism absent in both loss-level combination and learning-to-defer frameworks.

\section{A Bayesian Model of Human--AI Interaction}

Let $Y \in \mathcal{Y}$ denote a random variable representing the true state of the world and the underlying uncertainty faced by the decision-maker. The human's knowledge is modeled as a random variable $H \in \mathcal{H} $ that provides information about the state of the world $Y$. Similarly, the AI system produces a signal $A \in \mathcal{A} $, also modeled as a random variable informative about $Y$.

The task for the human is to choose a decision $d \in \mathcal{D}$ so as to minimize an expected loss that depends on the realized state of the world. The loss function is given by
\[
L : \mathcal{Y} \times \mathcal{D} \to \mathbb{R}.
\]
If the human were to observe the true state $Y$ directly, the optimal decision would be
\[
\delta(Y) = 
\arg\min_{d \in \mathcal{D}} L(Y, d).
\]

In practice, however, the state of the world is uncertain. The human does not observe $Y$ and instead has access only to their own information $H$ together with the AI signal $A$. After observing both signals, the human chooses a decision according to a decision rule
\[
\delta : \mathcal{H} \times \mathcal{A} \to \mathcal{D},
\]
which maps realizations of $H$ and $A$ into a decision.
The expected loss incurred by the AI-assisted human is therefore
\[
\mathbb{E}_{Y,H,A}\big[ L\big( Y, \delta(H,A) \big) \big].  %
\]
This loss is the central object of interest in our paper. Our goal is to understand which aspects of the problem—relating to the environment, the human decision-maker, the AI system, and their interaction—shape this loss, and thereby contribute to either the success or the failure of human--AI collaboration.

\begin{example}[AI-Assisted Clinical Diagnosis]\label{eg: clinical diagnosing}
Consider the task of diagnosing a patient for a particular disease. The state of the world $Y$ indicates whether the patient has the disease, and its prior distribution reflects the disease prevalence in the population. The loss function corresponds to a classification loss, capturing the costs associated with false positives and false negatives.

The human decision-maker is a physician whose knowledge consists of medical training, prior experience, and observations made during the consultation. This information, together with diagnostic tests and interactions with the patient, is summarized by the random variable $H$, which is correlated with the true disease state. For example, observing that a patient appears fatigued may increase the physician’s posterior belief that the disease is present.

The AI system has access to information summarized by the signal $A$. This includes medical textbooks in its training set, large-scale historical medical records and patterns learned from prior cases. The AI signal is informative about the disease state through its correlation with $Y$, for instance, by comparing the current patient to similar historical cases.

Both $H$ and $A$ can be interpreted as summary statistics of richer underlying information. In simple settings, they may correspond to probabilistic assessments—for example, the physician’s and the AI’s estimated probability that the patient has the disease, although they may also represent more detailed or structured forms of communication. The physician then combines these two signals to form a diagnosis, represented by the decision rule $\delta(H,A)$.
\end{example}

\section{The Value of Information}
Before analyzing human--AI interaction, we first introduce a principled way to compare the usefulness of
different signals: in particular, the value of the AI signal to the human decision-maker. We define the \emph{value of information} as a \emph{decision-theoretic} object:
it measures how much a signal could reduce expected loss if it were used \emph{optimally}. This serves as an evaluation of the fundamental value of the signals. In practice, performance of the AI-assisted human can be lower because humans may not use the signal optimally---they may misinterpret, overweight, or underweight AI recommendations. We will capture this aspect later in the paper.

\begin{definition}[Value of Information]  \label{def: value}
Let $S\in\cS$ be a signal (e.g., $H$, $A$, or the concatenation signal $(H,A)$). Define the no-information baseline Bayesian risk
\[
L_0:=\inf_{\delta_0\in \mathcal{D}} 
\mathbb{E}_{Y}\big[ L\big( Y, \delta_0 \big) \big].
\]
The value of information of $S$ is defined as
\[
v(S) := L_0-\inf_{\delta: \mathcal{S} \to \mathcal{D}} 
\mathbb{E}_{Y,S}\big[ L\big( Y, \delta(S) \big) \big]\geq 0.
\]
\end{definition}

The quantity $v(S)$ is the maximum achievable risk reduction from observing $S$, relative to the
best constant action. Nonnegativity follows because the decision-maker can always ignore the signal
and choose the baseline action. Notably, $v(S)$ is an upper bound on the benefit of access to $S$ under
any constrained or misspecified decision rule.

For a fixed task, successive generations of AI systems can be viewed as generating different AI signals $A_1,A_2,\ldots$ with different informativeness about $Y$. Performance improvement of AI models on standard \emph{static} benchmarks (e.g., a fixed classification task)
can therefore be interpreted as evidence of increasing $v(A)$: the standalone decision value of the AI
signal when used optimally. Importantly, however, for AI-assistance the relevant object is
typically not $v(A)$ but the \emph{joint} value $v(A,H)$, which we will analyze soon.

This definition of information value is built from the task (loss function) and the joint distribution of $(Y,S)$,
directly capturing \emph{decision relevance} of a signal. Familiar statistical quantities such as conditional variance and mutual
information naturally emerge when we specify the problem instance.

\begin{example}[Regression]\label{eg: regression}
Suppose the state is real-valued $Y\in\RR$ and the loss function is quadratic: $L(Y,d)=(Y-d)^2$. Then we have
\[
v(S)=\Var(Y)-\Var (Y\mid S).
\]
\end{example}

\begin{example}[Classification]
Suppose the state is binary $Y\in\{0,1\}$ and the loss function is log loss: $L(Y,d)=-Y\log d-(1-Y)\log (1-d)$. The value of a signal $S$ is
\[
v(S)=I(Y;S),
\]
where $I$ denotes the mutual information.
\end{example}

\subsection{Joint value and marginal value}
In our setting the human has access to $H$, and may additionally observe the AI signal $A$. Under
Bayesian-optimal inference, the best achievable loss by the AI-assisted human is:
\[
\inf_{\delta: \cH \times \cA \to \cD} \Eb[L(Y, \delta(H,A))]=L_0-v(H,A).
\]
where $v(H,A)$ is shorthand for the value of the concatenated signal $(H,A)$.

This identity clarifies an important point for model choice and deployment: even if two AI systems
have different standalone values $v(A)$, what matters for a fixed human decision-maker is how much
additional value the AI provides \emph{on top of} the human's information:
\begin{definition}[Marginal Value of Information]  \label{def:marg_info}
The marginal value of AI information $A$ complementing human information is defined as
\begin{align*}
v(A \mid H) &:=
\inf_{\delta: \mathcal{H} \to \mathcal{D}}
\mathbb{E}\left[ L(Y, \delta(H)) \right]
-
\inf_{\delta: \mathcal{H}\times\mathcal{A} \to \mathcal{D}}
\mathbb{E}\left[ L(Y, \delta(H,A)) \right]\\
&=v(A,H)-v(H)
\end{align*}
\end{definition}

The marginal value of AI, $v(A\mid H)$, isolates the incremental decision value of AI \textit{conditional on the human's
knowledge}. It is the appropriate measure for evaluating AI \emph{as assistance}, rather than AI \emph{in
isolation}. 

Interestingly, the marginal value of AI, $v(A\mid H)$, can be larger or smaller than its standalone value $v(A)$, depending
on whether the signals are informational complements or substitutes \citep{borgers2013signals}. Note that
\[
v(A\mid H)>v(A) \iff v(A,H)>v(A)+v(H),
\]
Thus, $v(A\mid H)$ exceeds $v(A)$ exactly when the joint value is \emph{superadditive}: the two signals
create synergy in decision value. Conversely, when $A$ largely replicates the evidence already in $H$,
we expect $v(A,H)\approx v(H)$, so that $v(A\mid H)$ is much smaller than $v(A)$. The following proposition shows that their ratio can be arbitrarily large.

\begin{proposition}
\label{prop:relative_complement}
For any $t\in[0,+\infty)$, there exists a problem instance such that
\[
\frac{v(A\mid H)}{v(A)}=t.
\]
\end{proposition}
	The ratio $t=v(A\mid H)/v(A)$ can be read as a measurement for relative complementarity.
	Here, we provide two concrete examples where $t$ can take extreme values in the medical context.
		\begin{itemize}
			    
				    \item (very large $t$) Suppose $A$ and $H$ are conditionally independent given the disease state $Y$ because they draw on different evidence sources: e.g., the AI aggregates large-scale historical patterns and laboratory test interpretation, while the clinician contributes bedside observations and contextual information not present in the data. If the disease is extremely rare (very low prior) and treatment requires surpassing a threshold diagnosis probability, then neither the human alone nor the AI alone can achieve a high enough diagnosis confidence to trigger treatment, despite high true positive and low false positive rates. In this case, $v(A)$ and $v(H)$ can be near zero. Yet, combining the two independent pieces of evidence can shift the posterior significantly, crossing the threshold, making $v(H,A)$ large and hence $v(A\mid H)$ large.
			    
				    For example, suppose the disease prevalence is $0.1\%$; the human-alone and AI-alone true positive rates are both $70\%$, and their false positive rates are both $1\%$. In this case, a single positive AI or human flag yields only a $6.5\%$ posterior confidence in the presence of the disease. However, when both flags are positive, the posterior confidence jumps to $83.1\%$, which is much more likely to trigger treatment.

			    \item (very small $t$) If the AI is trained on and primarily summarizes the same evidence the expert already uses (e.g., the same X-ray image), then $A$ is highly redundant relative to $H$. The AI may still have standalone value for non-experts ($v(A)>0$), but conditional on an expert's signal its marginal
		value is small ($v(A\mid H)\approx 0$).
		\end{itemize}
A key conceptual implication is that \textit{standalone AI performance benchmarks can be misleading
for AI-assistance to a human}: the same improvement in $v(A)$ can translate into very different changes
in $v(A\mid H)$ depending on the overlap structure between $A$ and $H$. In particular, improvements in standalone AI performance $v(A)$ may
come from becoming \emph{better at redundant information}, which has little impact on the marginal value $v(A\mid H)$. 
In other words, better AI (higher $v(A)$) does not necessarily lead to better AI-assisted human performance (higher $v(A\mid H)$); the relationship can be complex and non-monotone. Throughout the paper, we will study in depth how AI capability impacts  AI-assisted human performance and identify critical aspects that govern human-AI interaction.

\subsection{Misspecification and complementarity gap}
So far, the marginal value $v(A\mid H)$ is defined under \emph{optimal} usage of the AI signal $A$ by the human.
However, making the Bayesian-optimal decision requires complete knowledge of the joint distribution of $(Y,H,A)$: the AI's reliability, the
human's own expertise, and crucially the dependence structure between $H$ and $A$. In particular, it would require understanding when the two sources tend to agree or disagree, and in which circumstances each is more likely to be correct. Clearly, human beliefs are far from such detailed knowledge in most real settings.

We therefore model the human as using a (possibly ``misspecified'') decision rule $\hat{\delta}$, leading
to realized  AI-assisted human loss
\[\mathbb{E}_{Y,H,A}\big[ L\big( Y, \hat{\delta}(H,A) \big) \big].
\]

Different instances of $\hat{\delta}$ correspond to different human behavioral biases (e.g., correlation neglect, overtrust, limited attention). In Section \ref{subsec:correlation_neglect}, we will study in-depth one such example of human \textit{behavioral misspecification}.

The next theorem characterizes two forces determining the performance of the  AI-assisted human: the marginal value
of AI information and the degree of misspecification in how the human uses it.

\begin{theorem}[Complementarity Gap]\label{thm: complementarity gap}
Suppose $L$ is a Bregman loss (e.g., squared loss, cross-entropy). Let $\hat{\delta}$ be the decision rule of the imperfect human decision-maker, and let $\delta^\star$ be the Bayesian-optimal decision rule. Then
\[
\underbrace{\inf_{\delta: \mathcal{H} \to \mathcal{D}}\, \mathbb{E}_{Y,H}\big[L(Y, \delta(H))\big]}_{\text{Human Alone}}
-
\underbrace{\mathbb{E}_{Y,H,A}\big[L(Y, \hat{\delta}(H,A))\big]}_{\text{Realized Human-AI}}
=
\underbrace{v(A\mid H)}_{\text{Marginal Value of AI}}
-
\underbrace{\mathbb{E}_{H,A}\big[L(\delta^\star(H,A),\hat{\delta}(H,A))\big]}_{\text{Human Misspecification}}
\]
\end{theorem}
The identity decomposes the \emph{realized} impact of the human's use of the AI into two terms: the first term $v(A\mid H)$ captures the maximum improvement a Bayes-optimal decision-maker could extract
from observing the AI signal; the second term is a \emph{behavioral penalty}: under Bregman losses, it can be interpreted as an
expected divergence between the Bayes-optimal action and the action induced by the human's misspecified decision rule. 

Notice that the right-hand side can be negative, which means the AI assisted human is worse than the human deciding alone. This is the case when the human's missuse of the AI (captured by the behavioral penalty) overwhelms the marginal value of the AI signal.

A key remark is that these two terms are not independent, and their interaction can be subtle: an AI with higher marginal value can also be more severely misused by the human. In the following, we disentangle the key drivers behind these forces and analyze how they interact in shaping the AI-assisted human performance.

\section{A Conceptual Model of Overlap}\label{sec:conceptual}

So far, our framework has identified two forces that govern the
performance of an AI-assisted human. The first is informational: the marginal value of AI information $v(A\mid H)$, which is
determined by how much genuinely new information the AI provides beyond what the human already
knows. The second is behavioral: human missuse of AI,
which can prevent the human from optimally extracting that incremental AI information. 

Our goal now is to understand how these two forces interact in shaping the performance of the AI-assisted human across domains. We seek to address the following natural questions: When does introducing an AI signal improve human decision-making, and when can it harm it? When is automation preferable, and when does augmentation yield complementarity?

We study these questions from the perspective of a fixed human decision-maker with expertise $H$ and a (possibly misspecified) decision rule $\hat{\delta}$. Across domains, two critical features vary: AI capability (how good the AI model is at the domain's task) and ``informational overlap'' between $H$ and $A$ (how redundant the AI's information is relative to human expertise).
While AI capability can be measured through $v(A)$,
the notion of overlap is more subtle as it depend on the structure of the joint distribution of $H\mid Y$ and $A\mid Y$. This section is dedicated to the study of this critical notion.

Specifically, we introduce a micro-founded signal-construction model from which a measure of overlap naturally emerges. We assume that
the human signal $H$ and the AI signal $A$ are built by aggregating a common pool of primitive cues. This will make the \emph{overlap structure} between $H$ and $A$ explicit and measurable through an
\emph{overlap coefficient} $\lambda$.

\subsection{Model}\label{sec:conceptual_model}
Let the state $Y$ be normally distributed,
\[
Y \sim \cN(\mu_0,\tau_0^{-1}),
\]
and consider quadratic loss $L(Y,d)=(Y-d)^2$, so posterior variance directly maps into Bayesian risk.

Suppose the environment contains a collection of
conditionally independent Gaussian primitive cues,\footnote{The following analysis readily extends to heterogeneous primitive precisions $\{\tau_i\}$ as long as their variation is bounded, so that no single cue asymptotically dominates. See Appendix~\ref{appendix:hetero} for details.}%
\[
\cS=\{s_1,\ldots,s_N\},\qquad
s_i \perp\!\!\!\perp s_j \mid Y,\qquad
s_i \mid Y \sim \cN\left(Y,N\tau^{-1}\right).
\]
The scaling $N\tau^{-1}$ normalizes total information in the world: pooling all $N$ cues yields a
signal with precision $\tau$. 

The human observes a subset $\cH\subseteq \cS$ and forms the sufficient statistic (the sample mean)
\[
H := \frac{1}{|\cH|}\sum_{s_i\in\cH}s_i,
\qquad\text{so that}\qquad
H\mid Y \sim \cN\left(Y,\frac{N}{|\cH|}\tau^{-1}\right).
\]
Similarly, the AI observes $\cA\subseteq\cS$ and outputs
\[
A := \frac{1}{|\cA|}\sum_{s_i\in\cA}s_i,
\qquad\text{so that}\qquad
A\mid Y \sim \cN \left(Y,\frac{N}{|\cA|}\tau^{-1}\right).
\]
We summarize the standalone strength of the human and AI signals by
their conditional precisions
\[
\tau_H := \Var(H\mid Y)^{-1}=\frac{|\cH|}{N}\tau,
\qquad
\tau_A := \Var(A\mid Y)^{-1}=\frac{|\cA|}{N}\tau.
\]
In this construction, increasing AI capability corresponds to increasing $|\cA|$ (the number of
distinct cues, or data points, that the AI model effectively aggregates), which increases $\tau_A$ up to the environmental cap. In fact, $v(A)=\tau_0^{-1}-(\tau_0+\tau_A)^{-1}$, strictly increasing with $\tau_A$.

\subsection{Decomposition lemma and overlap coefficient}
The key domain-dependent object is how much of the AI's information is redundant given what
the human already knows. The next lemma formalizes this by decomposing the AI signal into
(i) a component predictable from the human signal, and (ii) an orthogonal innovation component.

\begin{lemma}[Decomposition Lemma]\label{lem:decomposition}
Conditional on $Y$, the pair $(H,A)$ is jointly Gaussian distributed. Define the \emph{overlap coefficient}
\[
\lambda
:=\frac{\Cov(H,A\mid Y)}{\Var(H\mid Y)}.
\]
Define the \emph{innovation signal}
\[
\tilde A := \frac{A-\lambda H}{1-\lambda}.
\]
Then we can write the decomposition
\[
A=\lambda H+(1-\lambda)\tilde A,
\]
where $\tilde A\mid Y\sim \cN(Y,\tilde\tau^{-1})$ and $\tilde A\perp\!\!\!\perp H\mid Y$, with precision
\[
\tilde\tau
=\left(\Var(\tilde A\mid Y)\right)^{-1}
=\frac{(1-\lambda)^2}{\tau_A^{-1}-\lambda^2\tau_H^{-1}}.
\]
Moreover, under the micro-founded cue construction,
\[
\lambda=\frac{|\cA\cap\cH|}{|\cA|}.
\]
\end{lemma}

Here, the coefficient $\lambda$ is the regression coefficient of $A$ on $H$ (conditional on $Y$), and in our
micro-foundation it has a sharp combinatorial meaning:
\[
\lambda=\frac{|\cA\cap\cH|}{|\cA|}
=\frac{\text{AI cues already known by the human}}{\text{total AI cues}}.
\]

Thus, the overlap coefficient $\lambda$ directly measures the share of the AI signal that is redundant \textit{from the human's
perspective}. Two extreme cases are:
\begin{itemize}
\item If $\lambda\approx 0$, then the AI mostly draws on cues the human does not have; $A$ is largely
novel information, and we expect strong potential for human-AI complementarity.
\item If $\lambda\approx 1$, then the AI is essentially re-aggregating what the human already knows. In this case, even a very accurate AI may add little marginal value, which could be easily overwhelmed by behavioral penalty.
\end{itemize}

The innovation signal $\tilde A$ is the fundamental additional information that the AI provides beyond what the human already knows. The decomposition shows that observing $(H,A)$ is equivalent to
observing $(H,\tilde A)$, where $H$ and $\tilde A$ are conditionally independent given $Y$. Therefore, the marginal value of AI information reduce to the precision of the orthogonal innovation signal, $\tilde\tau$.
Precisely, it takes a closed form:
\[
v(A\mid H)=\Var(Y\mid H)-\Var(Y\mid H,A)
=\frac{1}{\tau_0+\tau_H}-\frac{1}{\tau_0+\tau_H+\tilde\tau}.
\]
This highlights a key conceptual message: for the AI-assisted human performance, what matters is not $\tau_A$
per se, but the incremental precision $\tilde\tau$.

\subsection{Interpreting overlap coefficient}  \label{sec:sampling_two_worlds}
\newcommand{\cSacc}{\cS_{\text{ai-acc}}}
\newcommand{\cSnoacc}{\cS_{\text{ai-noacc}}}
Lemma~\ref{lem:decomposition} shows that the joint information structure of $(H,A\mid Y)$ is
summarized by two ingredients: (i) how strong the AI signal is (i.e., $\tau_A$ or equivalently $v(A)$), and (ii) how redundant it
is relative to the human (i.e., $\lambda$). These two parameters are sufficient to determine the incremental
precision $\tilde\tau$, and therefore pin down $v(A\mid H)$ and, under behavioral models, the AI-assisted human loss.

In practical settings, there are domains where AI draws primarily on information that human does not have, and there are also domains where AI is essentially re-aggregating human knowledge. We argue that, for a fix domain, and within a time period where the information source of AI is relatively stable (e.g., mostly from the web text), the overlap coefficient $\lambda$ remains stable across AI capability generations (e.g., from GPT-3 to GPT-4 to GPT-5). In particular, we instantiate the model in Section~\ref{sec:conceptual_model} on a setting where the AI randomly acquires new cues from its accessible pool (which is determined by the domain's information structure). We show that $\lambda$ is approximately invariant to increases in $|\cA|$ (equivalently, $\tau_A$, AI capability). This implies that the overlap coefficient $\lambda$ is primarily \textit{domain-specific}, while AI capability growth
primarily raises $\tau_A$ without affecting $\lambda$. 
This allow us to vary both AI capability $\tau_A$ and domain overlap $\lambda$ independently in analyzing how they shape the performance of the AI-assisted human. 

Formally, Suppose the primitive cue space splits into cues the AI can potentially access (e.g., historical records,
test results) and cues it cannot (e.g., certain real-time physical observations):
\[
\cS=\cS_{\text{ai-acc}}\;\cup\;\cS_{\text{ai-noacc}},
\qquad
\cS_{\text{ai-acc}}\cap \cS_{\text{ai-noacc}}=\emptyset.
\]
Let $|\cS_{\text{ai-acc}}|=m|\cS|$ and $|\cH\cap \cS_{\text{ai-acc}}|=k|\cS|$. For a given model
generation, the AI forms $\cA$ by sampling a subset of size $|\cA|=a|\cS|$ \emph{uniformly at random}
from $\cS_{\text{ai-acc}}$. We consider the regime where $|\cS|$ is large (information is sufficiently
granular), and we use the ratios $(a,m,k)$ to describe information-set sizes. AI capability growth is
captured by increasing $a$ from $0$ (no usable data) up to $m$ (exhaust all AI-accessible data). The setup is illustrated in Figure~\ref{fig:sampling_two_worlds}.

\begin{figure}

\centering

\usetikzlibrary{arrows.meta,calc,positioning,decorations.pathreplacing,patterns}

\newif\ifshowprop
\showpropfalse

\begin{tikzpicture}[x=1cm,y=1cm,>=Latex, font=\normalsize]

\def\lambdaLimit{\dfrac{k}{m}}

\definecolor{Hblue}{RGB}{40,120,220}
\definecolor{Ared}{RGB}{200,40,40}
\definecolor{Ppurple}{RGB}{140,60,230}

\def\W{11}     %
\def\H{7}      %
\def\xsplit{6} %

\def\notationY{\H*0.35}
\def\bottomLayerY{0.9}

\def\OuterBorderWidth{0.7pt}
\def\InnerBorderWidth{1.0pt}
\def\ArrowBorderWidth{0.5pt}

\def\RectAIacc{(0,0) rectangle (\xsplit,\H)}
\def\RectAInoacc{(\xsplit,0) rectangle (\W,\H)}
\def\EllipH{(\W/2, \H*0.75) ellipse [x radius=\W/2*0.9, y radius=1]}
\def\EllipA{(2, \H/2*0.9) ellipse [x radius=1, y radius=\H/2*0.8]}

\begin{scope}
    \clip \RectAIacc;
    \path[draw=Hblue, line width=\InnerBorderWidth,
            fill=Hblue!20] \EllipH;
\end{scope}

\node[Hblue, anchor=east] (Hcaplabel) at (\xsplit*0.95, \H*0.75) {$\cH \cap \cSacc$};

\begin{scope}
    \clip \RectAInoacc;
    \path[draw=Hblue, line width=\InnerBorderWidth, dashed]
        \EllipH;
\end{scope}

\node[Hblue] (Hlabel) at (\W*0.9, \H*0.88) {$\cH$};

\draw[Ared, line width=\InnerBorderWidth, fill=Ared!15] \EllipA;
\node[Ared] (Alabel) at (3, \bottomLayerY) {$\cA$};

\begin{scope}
    \clip \EllipH;
    \path[draw=none, line width=\InnerBorderWidth, 
            fill=Ppurple!35]
        \EllipA;
\end{scope}

\begin{scope}
    \clip \RectAIacc;
    \path[draw=Hblue, line width=\InnerBorderWidth,
            fill=none] \EllipH;
\end{scope}
\draw[Ared, line width=\InnerBorderWidth] \EllipA;

\node[anchor=east,
      draw, dashed, rounded corners=3pt, inner sep=5pt]
    (lamfrac) at (\xsplit*0.95, \notationY) 
    {\small $\dfrac{\color{Ppurple} |\cA\cap \cH|}{\color{Ared} |\cA|}=\lambda$};
\node (lamfrac num) at ([xshift=-0.3cm] lamfrac.north) {};
\node (lamfrac denom) at ([xshift=-0.3cm] lamfrac.south) {};

\path [->, Ppurple, line width=\ArrowBorderWidth]
    (2, \H*0.73) edge [bend left] (lamfrac num);

\path [->, Ared, line width=\ArrowBorderWidth]
    ([xshift=0.2cm] Alabel.east) edge [bend right] (lamfrac denom);

\ifshowprop
    \node[anchor=east,
          draw, dashed, rounded corners=3pt, inner sep=5pt] 
        (kmfrac) at ({\xsplit + (\W-\xsplit) * 0.9}, \notationY)
        {\small $\dfrac{\color{Hblue} |\cH \cap \cSacc|} {|\cSacc|}=\lambdaLimit$};
    \node (kmfrac num) at ([xshift=-0.4cm] kmfrac.north) {};
    \node (kmfrac denom) at ([xshift=-0.4cm] kmfrac.south) {};
    
    \path [->, Hblue, line width=\ArrowBorderWidth]
        (\xsplit, \H*0.75) edge [bend left] (kmfrac num);
        
    \path [->, black, line width=\ArrowBorderWidth]
        (\xsplit, \bottomLayerY) edge [bend right=15] (kmfrac denom);

    \path [->, black, line width=\ArrowBorderWidth]
        ([xshift=0.2cm] lamfrac.east) edge node [below] {\footnotesize $|\cS|\to+\infty$} ([xshift=-0.2cm] kmfrac.west);
\fi

\draw[line width=\OuterBorderWidth] \RectAIacc;
\node [above=1.8pt] at (\xsplit/2, \H) {Cues accessible to AI};

\draw[dashed, line width=\OuterBorderWidth] \RectAInoacc;
\node [above=1.8pt] at ({(\W+\xsplit)/2}, \H) {Cues inaccessible to AI};

\node[anchor=south east, inner sep=5pt] at (\xsplit, 0) {$\cSacc$};
\node[anchor=south east, inner sep=5pt] at (\W, 0) {$\cSnoacc$};

\end{tikzpicture}

\ifshowprop
    \caption{Conceptual model with cues split into two subsets (Section~\ref{sec:sampling_two_worlds}). The AI can only access primitive cues in $\cSacc$ and samples uniformly from this set.
    }
\else
    \caption{Conceptual model with cues split into two subsets (Section~\ref{sec:sampling_two_worlds}). The AI can only access primitive cues in $\cSacc$ and samples uniformly from this set.
    }
\fi
\label{fig:sampling_two_worlds}

\end{figure}

The ratio
\[
\frac{|\cH\cap \cS_{\text{ai-acc}}|}{|\cS_{\text{ai-acc}}|}=\frac{k}{m}
\]
is a \emph{domain-specific} overlap rate: it is the fraction of all AI-accessible cues that the human
already has in their information set.

\begin{proposition}[Overlap stabilization under random sampling]\label{prop:lambda_concentration}
Let $\cA$ be drawn uniformly among all $a|\cS|$-element subsets of $\cS_{\text{ai-acc}}$. Then the overlap
coefficient $\lambda=|\cA\cap\cH|/|\cA|$ satisfies
\[
\lambda \overset{p}{\rightarrow} \frac{k}{m} \quad \text{as}\quad |\cS|\rightarrow +\infty .
\]
where ``$\overset{p}{\rightarrow}$'' denotes convergence in probability.
\end{proposition}

Proposition~\ref{prop:lambda_concentration} highlights that, in this granular-information regime, $\lambda$ is
approximately pinned down by the domain parameter $k/m$ and is thus \emph{invariant} to AI
capability growth (changes in $a$). Intuitively, as the AI draws more cues from the same accessible
pool, it acquires redundant and non-redundant information in roughly fixed proportions.

This decoupling allows us to compare different AI models across domains and generations:
\begin{itemize}
\item \textbf{Across domains:} As captured by $\lambda$, domain-specific overlap can vary substantially. Domains where humans and AI rely on similar underlying evidence (high $k/m$) will exhibit high overlap,
whereas domains where humans primarily bring distinct, hard-to-digitize information (low $k/m$) can exhibit low
overlap and greater potential for complementarity.
\item \textbf{Across AI model generations:} AI capability improvements mainly
raise $\tau_A$ by allowing the AI to aggregate more cues (larger $a$), while leaving $\lambda$ roughly
stable. 
\end{itemize}
Finally, note the limit of this invariance claim: $\lambda$ \emph{can} change if the AI-accessible pool
itself expands (e.g., new modalities, new sensors) or if human workflows
shift (human reallocates their attention, changing which AI-accessible cues enter human information set $\cH$). Our perspective here is that conditional on a
stable access regime, $\lambda$ is best viewed as a \emph{domain parameter}, whereas $\tau_A$ is the natural
index of AI capability growth. This motivates studying joint performance as a function of $(\tau_A,\lambda)$
in the next section, leading to our main insights.

\section{ Human--AI Interaction Analysis}

In this section, we study how the Human--AI interaction evolves across different domains ($\lambda$) and AI capabilities $\tau_A$. We characterize when different \emph{interaction regimes} occur such as \textit{augmentation}, \textit{automation} and \textit{complementarity}. 
We study these regimes in two different human decision-making models. First,
a Bayes-rational human, to isolate the pure information
channel. 
Second, a correlation-neglect human, which captures a natural form of behavioral misspecification. The contrast between these two benchmarks reveals when augmentation helps, when it harms, and when automation dominates.

Recall that Section~\ref{sec:conceptual} provides a micro-founded Gaussian instance for the joint distribution
of $(Y,H,A)$, summarized by three parameters:
\begin{align*}
\tau_H &= \Var(H\mid Y)^{-1} &&\text{(human signal precision)},\\
\tau_A &= \Var(A\mid Y)^{-1} &&\text{(AI signal precision)},\\
\lambda &=\frac{\Cov(H,A\mid Y)}{\Var(H\mid Y)} &&\text{(overlap coefficient)},
\end{align*}
satisfying the natural constraint implied by the micro-founded model:
\[
\lambda =\frac{|\cA\cap \cH|}{|\cA|}\leq \min \left\{\frac{|\cH|}{|\cA|},1\right\}=\min\left\{\frac{\tau_H}{\tau_A},1\right\}.
\]
These quantities fully determine the conditional joint distribution of $(H,A)\mid Y$.
By Lemma~\ref{lem:decomposition}, we can orthogonalize the AI signal as
\[
A=\lambda H+(1-\lambda)\tilde A,
\qquad
\tilde A\perp\!\!\!\perp H\mid Y,
\]
where the incremental precision is
\[
\tilde\tau
:=\Var(\tilde A\mid Y)^{-1}
=\frac{(1-\lambda)^2}{\tau_A^{-1}-\lambda^2\tau_H^{-1}}.
\]

\subsection{Perfect (Bayesian-rational) Human Decision-maker}\label{sec:Bayes-opt_Human}

In this benchmark, the decision-maker knows the true joint distribution of $(Y,H,A)$ and implements the Bayes-optimal decision rule. She uses the AI signal optimally and extracts the innovation signal $\tilde A$, accounting for redundancy. As a result, access to AI cannot worsen performance: \emph{complementarity is always achieved under Bayesian-optimal integration}.

Here, when the human acts alone, the loss is
\[
L_{\mathrm{H}} := \inf_{\delta: \mathcal{H} \to \mathcal{D}} \mathbb{E}\left[L(Y, \delta(H))\right]
= \Var(Y\mid H)=\frac{1}{\tau_0+\tau_H}.
\]
When the human uses AI, the joint loss becomes
\[
L_{\mathrm{H-AI}}
:= \inf_{\delta: \mathcal{H}\times \mathcal{A} \to \mathcal{D}} \mathbb{E}\left[L(Y, \delta(H,A))\right]
= \Var(Y\mid H,A)
= \Var(Y\mid H,\tilde A)
= \frac{1}{\tau_0+\tau_H+\tilde\tau}.
\]
Since there is no behavioral misspecification, the realized gain from AI access is exactly the marginal value of AI information:
\begin{align*}
v(A\mid H)
&=L_{\mathrm{H}}-\widehat L_{\mathrm{H-AI}}.
\end{align*}
The following proposition describes how the impact of AI on human decision-making varies with (i) AI standalone capability, (ii) overlap with human information, and (iii) human expertise.

\begin{proposition}[Comparative statics]
\label{prop:comparative_statics_bayes}
Consider feasible pairs of  $(\tau_A,\tau_H, \lambda)$ satisfying $0<\lambda<\min\left\{\frac{\tau_H}{\tau_A},1\right\}$.
Then:
{\renewcommand{\labelenumi}{(\roman{enumi})}%
		\begin{enumerate}[(i)]
  \item $v(A\mid H)$ is strictly increasing in $\tau_A$.
  \item $v(A\mid H)$ is strictly decreasing in $\lambda$.
  \item $v(A\mid H)$ is strictly decreasing in $\tau_H$.
\end{enumerate}
}
\end{proposition}

Proposition~\ref{prop:comparative_statics_bayes} isolates the \emph{pure informational channel}.
It says that if humans could integrate AI optimally, then (i) improving AI capability always helps,
(ii) overlap is unambiguously harmful because it reduces the incremental precision $\tilde\tau$, and (iii) the more
expert the human (higher $\tau_H$), the less marginal benefit any AI signal provides, featuring
diminishing returns of information.

Notably, any empirical failures of complementarity must be driven by the second force in our framework: behavioral misspecification. In the next section, we show that once the human does not use AI optimally, more subtle interactions emerge in how these factors shape the performance of the AI-assisted human.

\subsection{Correlation-neglect Human Decision-maker}
\label{subsec:correlation_neglect}
We now study the Human-AI interaction under human behavioral misspecification.
We consider the instance of a human who exhibits \emph{correlation neglect}: she treats her signal $H$ and the AI signal $A$ as conditionally independent given $Y$, even though they are not.
This is a natural stylized model of ``na\"ive evidence aggregation,'' in which agreement between sources
is mistakenly interpreted as independent corroboration. Empirical evidence of such behaviors includes \cite{eyster2010correlation,kallir2009neglect,enke2019correlation,Agarwal2023NBERradiology}.

Formally, we assume the \emph{subjective} model of the correlation-neglect human is
\begin{align*}
Y &\sim \cN(\mu_0,\tau_0^{-1}),\\
H\mid Y &\sim \cN\left(Y,\tau_H^{-1}\right),\\
A\mid Y &\sim \cN\left(Y,\tau_A^{-1}\right),\\
H&\perp\!\!\!\perp A\mid Y,
\end{align*}
so that the human correctly understands the prior distribution of state $Y$, the marginal distributions $A\mid Y$ and $H\mid Y$, but incorrectly believes that\footnote{Similar modeling choice of correlation neglect appears in \cite{ortoleva2015overconfidence,levy2015correlation,levy2022persuasion}.} $H\perp\!\!\!\perp A\mid Y$.  Upon receiving signals $H$ and $A$, the (subjective) posterior belief of a human decision maker with correlation neglect is
\[
Y\mid H,A \ \sim \ \cN\left(
\frac{\tau_0\mu_0+\tau_H H+\tau_A A}{\tau_0+\tau_H+\tau_A},\ \frac{1}{\tau_0+\tau_H+\tau_A}
\right).
\]

The correlation-neglect human's decision rule is
\[{\delta}^{\mathrm{CN}}(H,A) = \arg\min_{d \in \mathcal{D}} \mathbb{E}^{\text{CN}}\big[ L(Y, d) \mid H,A\big],
\]
where $\mathbb{E}^{\mathrm{CN}}[\cdot]$ denotes expectation computed under the subjective model.
Under quadratic loss, the decision rule is simply the posterior mean under the subjective model
\begin{equation}
\delta^{\mathrm{CN}}(H,A)
= \mathbb{E}^{\mathrm{CN}}[Y\mid H,A]
= \frac{\tau_0\mu_0+\tau_H H+\tau_A A}{\tau_0+\tau_H+\tau_A}.
\label{eq:cn_rule}
\end{equation}

Substituting $A=\lambda H+(1-\lambda)\tilde A$ into~\eqref{eq:cn_rule} yields
\begin{equation*}
\delta^{\mathrm{CN}}(H,A)
= \frac{\tau_0}{T}\mu_0
+ \frac{\tau_H+\lambda\tau_A}{T}\,H
+ \frac{(1-\lambda)\tau_A}{T}\,\tilde A,
\qquad
T:=\tau_0+\tau_H+\tau_A.
\end{equation*}
The resulting expected loss of the AI assisted human admits a closed form. First,
\begin{align*}
\mathbb{E}\left[L \left( Y, {\delta}^{\mathrm{CN}}(H,A) \right)\right]
&=
\left(\frac{\tau_0}{T}\right)^2\frac1{\tau_0}
+
\left(\frac{\tau_H+\lambda\tau_A}{T}\right)^2\frac1{\tau_H}
+
\left(\frac{(1-\lambda)\tau_A}{T}\right)^2\frac1{\tilde\tau}.
\end{align*}
Plugging in the expression for $\tilde\tau$ simplifies this to
\begin{equation}
\widehat L_{\mathrm{H-AI}}
:=
\frac{1}{T}
+\frac{2\lambda\tau_A}{T^2},
\qquad T:=\tau_0+\tau_H+\tau_A.
\label{eq:cn_loss_closed_form}
\end{equation}

The decomposition in~\eqref{eq:cn_loss_closed_form} is interpretable:
\begin{itemize}
\item The term $1/T$ is the loss the decision-maker \emph{would} achieve if independence were true
(i.e., if she were correctly aggregating three independent Gaussian signals with
precisions $\tau_0,\tau_H,\tau_A$).
\item The additional term $2\lambda\tau_A/T^2$ is an \emph{overlap penalty}. It captures two forces: (i) reduced marginal value of AI information due to overlap, and (ii) double-counting of overlapping evidence by the correlation-neglect human. When $\lambda=0$, both forces vanish: the human's independence assumption is correct and correlation neglect is costless.
\end{itemize}

Our goal in what follow is to characterize the human-AI interaction across different domains ($\lambda$) and AI capabilities ($\tau_A$). Different regimes of interaction will arise from comparing the AI-assisted human loss, with the human only and AI only losses:

\begin{equation}
L_{\mathrm{H}}=\frac{1}{\tau_0+\tau_H}
\quad\text{(human-only)},\qquad
L_{\mathrm{AI}}=\frac{1}{\tau_0+\tau_A}
\quad\text{(AI-only)}.
\label{eq:baselines_cn}
\end{equation}

Key questions we address next are: when does \textit{augmentation} arise? when does \textit{automation} dominate? When is \textit{complementarity} achievable?

\subsubsection{Augmentation}
\begin{center}
	\textit{When does AI assistance improve human decision-making?}
\end{center}

Formally, we ask for which domains (overlap $\lambda$) and AI capability levels ($\tau_A$) the joint AI-assisted human loss $\widehat L_{\mathrm{H-AI}}$ is lower than the human-only loss $L_{\mathrm{H}}$. We refer to this regime as \emph{augmentation}.

A first-order expansion of the Human-AI loss around
$\tau_A=0$ yields
\[
\widehat L_{\mathrm{H-AI}}
=
L_{\mathrm{H}}
+\frac{(2\lambda-1)\tau_A}{(\tau_0+\tau_H)^2}
+o(\tau_A).
\]
The expansion shows that the initial effect of introducing a weak AI is beneficial when $\lambda<\tfrac12$ and negative when
$\lambda>\tfrac12$. This hints that the overlap cutoff $\lambda=\tfrac12$ plays a central role in determining whether introducing AI into the human workflow is beneficial.

The next proposition provides an explicit characterization as to when AI assistance benefits the human and when it deteriorates her performance.

\begin{proposition}[Augmentation]\label{prop:withhold_threshold}
Define the augmentation threshold
\begin{equation*}
\tau_{\rm aug}(\lambda)
=
(\tau_0+\tau_H)(2\lambda-1).
\label{eq:tau_H_to_CN}
\end{equation*}
\begin{enumerate}[(i)]
    \item If $\lambda< \tfrac12$, then $ \widehat L_{\mathrm{H-AI}}< L_{\mathrm{H}}$ for all $\tau_A>0$.
    \item If $\lambda>\tfrac12$, then $\widehat L_{\mathrm{H-AI}}< L_{\mathrm{H}}$ if and only if $\tau_A> \tau_{\rm aug}(\lambda)$.
\end{enumerate}
\end{proposition}
Proposition~\ref{prop:withhold_threshold} shows that when $\lambda<\tfrac12$, AI always improves human decision-making, no matter AI capabilities: augmentation is always achieved. Intuitively, low overlap ($\lambda < \frac{1}{2}$) means the AI contributes more novel than redundant information as $\lambda = \frac{|\mathcal{A} \cap \mathcal{H}|}{|\mathcal{A}|}$ from the micro founded model. 

Its informational benefit dominates the behavioral penalty from correlation neglect.

However, when overlap is large ($\lambda>\tfrac12$), Proposition~\ref{prop:withhold_threshold} shows that
introducing an AI with low capabilities deteriorates performance of the human decision-maker.
A weak AI adds little novel information and the human missuse of it dominates the information benefit.
The threshold $\tau_{\rm aug}(\lambda)$ is increasing in $\lambda$: more overlap requires a stronger AI
before it is beneficial to use the signal to assist a correlation-neglect decision-maker.

Our model offers a prediction for early-stage AI deployment: when AI capability is limited and overlap with human expertise is high, introducing AI assistance can degrade performance relative to human-only decision-making. Augmentation becomes beneficial only after AI capability surpasses the threshold $\tau_{\mathrm {aug}}(\tau)$, which is higher in domains with greater informational overlap.

\subsubsection{Automation}

\begin{center}
		\textit{When is it optimal to rely on the AI alone rather than AI-assisted human or human alone?}
\end{center}
The following proposition characterizes the automation regime.
\begin{proposition}[Automation]
		\label{prop:automation}
		Define the automation threshold 
			\begin{equation*}
			\tau_{\mathrm{auto}}(\lambda)
		=
		\frac{\tau_H-2\lambda\tau_0+\sqrt{(\tau_H-2\lambda\tau_0)^2+8\lambda\tau_H(\tau_0+\tau_H)}}{4\lambda}
		>0.
		\label{eq:tau_CN_to_AI}
			\end{equation*}
		Then $L_{\mathrm{AI}} < \widehat L_{\mathrm{H-AI}} $ if and only if $\tau_A > \tau_{\mathrm{auto}}(\lambda)$. Here, $\tau_{\mathrm{auto}}(\lambda)$ is strictly decreasing in $\lambda$.\\
		Furthermore, define
			\[
			\bar\lambda
		=
		\frac12+\frac{\tau_H}{2(\tau_0+\tau_H)}
		\in\Big(\tfrac12,1\Big).
		\]
		\begin{enumerate}[(i)]
		\item If $\lambda < \bar\lambda$, then $L_{\mathrm{AI}} < \min\{L_{\mathrm{H}}, \widehat L_{\mathrm{H-AI}}\}$ if and only if $\tau_A > \tau_{\mathrm{auto}}(\lambda)$. \\
		In this case, we also have $\tau_{\mathrm{auto}}(\lambda) > \tau_H$.
		\item If $\lambda > \bar\lambda$, then $L_{\mathrm{AI}} < \min\{L_{\mathrm{H}}, \widehat L_{\mathrm{H-AI}}\}$ if and only if $\tau_A > \tau_H$.\\In particular, there is no range of AI capabilities for which the joint system outperforms both baselines: the complementarity region is empty.
	\end{enumerate}
	\end{proposition}

	Proposition~\ref{prop:automation} shows that \textit{automation is inevitable} with increasing AI capabilities: under human behavioral misspecification (correlation neglect here), there exists a threshold on AI capabilities beyond which relying on AI alone outperforms both human-only and AI-assisted human decision-making.
    
	Notably, under Bayesian-optimal integration (Section~\ref{sec:Bayes-opt_Human}), automation is never optimal---the human always adds value by correctly extracting the innovation signal. In our setting, the automation regime arises entirely from behavioral misspecification. This highlights a key implication: reducing correlation neglect (whether through training, interface design, or information structures that separate novel from redundant AI information) would raise the automation threshold, expanding the range of AI capabilities over which augmentation remains beneficial.

	When overlap is moderate (case (i) in Proposition~\ref{prop:automation}), the automation threshold ($\tau_{\mathrm{auto}}(\lambda)$) is above human capabilities ($\tau_H$). That is, automation requires more than the AI merely matching human capabilities. This is because the AI-assisted human still provides better performance than the human alone.
		
	When overlap is severe (case (ii) in Proposition~\ref{prop:automation}), automation becomes optimal as soon as AI capability exceeds human capability. Intuitively, severe overlap means that human and AI information is largely redundant, so AI assistance adds little new information while behavioral misspecification in human use of AI can dominate and reduce overall performance.

\subsubsection{Complementarity} 
\begin{center}
	\textit{When does the AI-assisted human improve upon both the AI alone and the human alone?}
\end{center}
The following proposition characterizes when complementarity arises.
\begin{proposition}[Complementarity]\label{prop:complementarity}
		Complementarity (i.e., $\widehat L_{\mathrm{H-AI}} < \min\{L_{\mathrm{AI}}, L_{\mathrm{H}}\}$) is possible if and only if
			\[
			\tau_{\rm aug}(\lambda)<\tau_{\mathrm{auto}}(\lambda),
		\]
	which is equivalent to the overlap condition
		\[
		\lambda<\bar\lambda.
		\]
	\begin{enumerate}[(i)]
		\item If $\lambda < \tfrac12$, then $\widehat L_{\mathrm{H-AI}} < \min\{L_{\mathrm{AI}}, L_{\mathrm{H}}\}$ if and only if $\tau_A \in (0, \tau_{\mathrm{auto}}(\lambda))$.
			\item If $\tfrac12< \lambda < \bar{\lambda}$, then $\widehat L_{\mathrm{H-AI}} < \min\{L_{\mathrm{AI}}, L_{\mathrm{H}}\}$ if and only if $\tau_A \in (\tau_{\rm aug}(\lambda), \tau_{\mathrm{auto}}(\lambda))$.
	\end{enumerate}
\end{proposition}

Proposition~\ref{prop:complementarity} shows that
complementarity is possible only when, as AI capability grows, augmentation precedes automation; i.e., the AI becomes strong enough to help ($\tau_A>\tau_{\rm aug}(\lambda)$) before it becomes strong enough to replace ($\tau_A>\tau_{\mathrm{auto}}(\lambda)$).
This is possible when the domain overlap is not too high ($\lambda<\bar\lambda$) so that the behavioral penalty from correlation neglect does not overwhelm the informational gain from AI assistance.

When overlap is low ($\lambda< \tfrac12$), complementarity is achieved for all AI capabilities until the automation threshold. 
When overlap is moderate ($\tfrac12<\lambda<\bar\lambda$), complementarity is achieved in an interval, between the augmentation and automation thresholds $\tau_A\in[\tau_{\rm aug}(\lambda),\tau_{\mathrm{auto}}(\lambda)]$.
Outside of this range, weak AI assistance penalizes the human as its misuse dominates its limited benefits, while very strong AI makes automation optimal as its misuse make it preferable to bypass the human.
As $\lambda$ increases, $\tau_{\mathrm{auto}}(\lambda)$ decreases while $\tau_{\rm aug}(\lambda)$ increases, narrowing the range of AI capabilities for which complementarity is achievable.

\subsubsection{Phase Transitions}
We are now able to give a complete characterization of the interaction regimes and a how the human-AI interaction evolves across different domains with increasing AI capabilities. We obtain three qualitatively different patterns, indexed by overlap.
Figure~\ref{fig:loss-three-panels} shows how expected losses varies with AI capability $\tau_A$ for each overlap region discussed below.

\def\tauzero{1}
\def\tauH{1}

\def\lambdalow{0.45}
\def\lambdamid{0.67}
\def\lambdahigh{0.85}

\def\xmaxLambdaLow{2.4}
\def\xmaxLambdaMid{1.7}
\def\xmaxLambdaHigh{1.7}

\def\yminLambdaLow{0.24}
\def\yminLambdaMid{0.33}
\def\yminLambdaHigh{0.33}

\def\ymaxLambdaLow{0.57}
\def\ymaxLambdaMid{0.57}
\def\ymaxLambdaHigh{0.58}

\newif\ifshowlegend
\showlegendfalse

\def\lossLabelLocation{right}  %

\def\optimalActionIndicator{bold}  %

\pgfplotsset{
    jointloss/.style={
        line width=\MediumCurveWidth,
        \jointcolor
    },
    humanloss/.style={
        line width=\MediumCurveWidth,
        \humancolor,
    },
    ailoss/.style={
        line width=\MediumCurveWidth,
        \aicolor
    },
}

\tikzset{
    xaxis regime labels/.style={
        anchor=south, 
        font=\scriptsize\bfseries,  %
        yshift=2.5pt
    },
    loss curve labels/.style={
        anchor=north, 
        font=\scriptsize,  %
        inner sep=4pt
    }
}

\def\HumanLossFunc{1/(\tauzero+\tauH)}
\def\AILossFunc{1/(\tauzero + x)}
\def\JointLossFunc{%
    1/(\tauzero+\tauH+x)%
    + 2*\lambdaval*x/((\tauzero+\tauH+x)^2)%
}

\usepgfplotslibrary{fillbetween}

\newif\ifLP@showCurveLabels
\newif\ifLP@showLegend

\pgfkeys{
  /lossplot/.is family, /lossplot,
  default/.style={
    lambdaval=\lambdamid,
    width=10cm,
    height=7cm,
    xmin=0,
    xmax=1.7,
    ymin=0.33,
    ymax=0.57,
    showCurveLabels=true,
    showLegend=false,
    regime=intermediate, %
    lossLabelLocation=\lossLabelLocation,
    optimalIndicator=\optimalActionIndicator,
  },
  lambdaval/.store in=\LP@lambdaval,
  width/.store in=\LP@width,
  height/.store in=\LP@height,
  xmin/.store in=\LP@xmin,
  xmax/.store in=\LP@xmax,
  ymin/.store in=\LP@ymin,
  ymax/.store in=\LP@ymax,
  showCurveLabels/.is if=LP@showCurveLabels,
  showLegend/.is if=LP@showLegend,
  regime/.store in=\LP@regime,
  lossLabelLocation/.store in=\LP@lossLabelLocation,
  optimalIndicator/.store in=\LP@optimalIndicator,
}

\def\LP@reg@intermediate{intermediate}
\def\LP@reg@low{low}
\def\LP@reg@high{high}
\def\LP@lossLabelLoc@on{on}%
\def\LP@lossLabelLoc@right{right}%
\def\LP@opt@shadeLower{shade lower}
\def\LP@opt@shadeVert{shade vert}
\def\LP@opt@bold{bold}

\newcommand{\DrawLossPanel}[1]{%
  \pgfkeys{/lossplot,default,#1}%
  \def\lambdaval{\LP@lambdaval}%
  \def\xmin{\LP@xmin}%
  \def\xmax{\LP@xmax}%
  \def\ymin{\LP@ymin}%
  \def\ymax{\LP@ymax}%
  \pgfmathsetmacro{\c}{\tauzero+\tauH}%
  \pgfmathsetmacro{\tauHtoHA}{\c*(2*\lambdaval - 1)}%
  \pgfmathsetmacro{\tauHAtoA}{%
    (\tauH - 2*\lambdaval*\tauzero + sqrt(8*\tauH*\tauH*\lambdaval + \tauH*\tauH + 4*\tauH*\lambdaval*\tauzero + 4*\lambdaval*\lambdaval*\tauzero*\tauzero))
    /(4*\lambdaval)%
  }%
  \pgfmathsetmacro{\tauHtoA}{\tauH}%
  \pgfmathsetmacro{\yH}{\HumanLossFunc}%
  \pgfmathsetmacro{\yAatTauHAtoA}{1/(\tauzero+\tauHAtoA)}%
  \edef\LP@regime@tmp{\LP@regime}%
  \ifx\LP@regime@tmp\LP@reg@intermediate
    \pgfmathsetmacro{\xWithhold}{(\xmin + \tauHtoHA)/2}%
    \pgfmathsetmacro{\xAugment}{(\tauHtoHA + \tauHAtoA)/2}%
    \pgfmathsetmacro{\xAutomate}{(\tauHAtoA + \xmax)/2}%
  \fi
  \ifx\LP@regime@tmp\LP@reg@low
    \pgfmathsetmacro{\xWithhold}{0}%
    \pgfmathsetmacro{\xAugment}{(\xmin + \tauHAtoA)/2}%
    \pgfmathsetmacro{\xAutomate}{(\tauHAtoA + \xmax)/2}%
  \fi
  \ifx\LP@regime@tmp\LP@reg@high
    \pgfmathsetmacro{\xWithhold}{(\xmin + \tauHtoA)/2}%
    \pgfmathsetmacro{\xAugment}{(\tauHtoHA + \tauHAtoA)/2}%
    \pgfmathsetmacro{\xAutomate}{(\tauHtoA + \xmax)/2}%
  \fi
  \pgfmathsetmacro{\yHumanMid}{\HumanLossFunc}%
  \pgfmathsetmacro{\yAIMid}{1/(\tauzero + \xAutomate)}%
  \pgfmathsetmacro{\yJointMid}{%
    1/(\tauzero+\tauH+\xAugment)%
    + 2*\lambdaval*\xAugment/((\tauzero+\tauH+\xAugment)^2)%
  }%
  \pgfmathsetmacro{\xLossLabelRight}{\xmax - 0.03*(\xmax-\xmin)}%
  \pgfmathsetmacro{\yHumanRight}{\HumanLossFunc}%
  \pgfmathsetmacro{\yAIRight}{1/(\tauzero + \xLossLabelRight)}%
  \pgfmathsetmacro{\yJointRight}{%
      1/(\tauzero+\tauH+\xLossLabelRight)%
      + 2*\lambdaval*\xLossLabelRight/((\tauzero+\tauH+\xLossLabelRight)^2)%
  }%
  \def\LP@xtick{}%
  \def\LP@xticklabels{}%
  \ifx\LP@regime@tmp\LP@reg@intermediate
    \def\LP@xtick{0, \tauHtoHA,\tauHAtoA}%
    \def\LP@xticklabels{$0$, $\HtoHAlabel$, $\HAtoAlabel$}%
  \fi
  \ifx\LP@regime@tmp\LP@reg@low
    \def\LP@xtick{0, \tauHAtoA}%
    \def\LP@xticklabels{$0$, $\HAtoAlabel$}%
  \fi
  \ifx\LP@regime@tmp\LP@reg@high
    \def\LP@xtick{0, \tauHtoA}%
    \def\LP@xticklabels{$0$, $\HtoAlabel$}%
  \fi
  \begin{tikzpicture}[trim axis left, trim axis right]%
    \begin{axis}[%
      width=\LP@width,%
      height=\LP@height,%
      xlabel={\AbilityAxisLabel},%
      ylabel={\LossAxisLabel},%
      xmin=\xmin, xmax=\xmax,%
      ymin=\ymin, ymax=\ymax,%
      domain=\xmin:\xmax,%
      samples=200,%
      axis lines=left,%
      legend style={%
        at={(0.97,0.97)},%
        anchor=north east,%
        draw=none,%
        fill=none%
      },%
      ticklabel style={font=\small},%
      label style={font=\small},%
      xtick/.expand once={\LP@xtick},%
      xticklabels/.expand once={\LP@xticklabels},%
      ytick=\empty,%
      set layers,%
      axis on top,%
      clip=true,%
    ]%
    \edef\LP@opt@tmp{\LP@optimalIndicator}%
    \ifx\LP@opt@tmp\LP@opt@bold
    \else
        \ifx\LP@opt@tmp\LP@opt@shadeLower
          \ifx\LP@regime@tmp\LP@reg@intermediate
            \addplot[name path=baseL, domain=\xmin:\tauHtoHA, draw=none] {\ymin};
            \addplot[name path=baseM, domain=\tauHtoHA:\tauHAtoA, draw=none] {\ymin};
            \addplot[name path=baseR, domain=\tauHAtoA:\xmax, draw=none] {\ymin};
            \addplot[name path=humanL, domain=\xmin:\tauHtoHA, draw=none] {\HumanLossFunc};
            \addplot[name path=jointM, domain=\tauHtoHA:\tauHAtoA, draw=none] {\JointLossFunc};
            \addplot[name path=aiR, domain=\tauHAtoA:\xmax, draw=none] {\AILossFunc};
            \addplot[humanshade] fill between[of=baseL and humanL];
            \addplot[jointshade] fill between[of=baseM and jointM];
            \addplot[aishade] fill between[of=baseR and aiR];
          \fi
          \ifx\LP@regime@tmp\LP@reg@low
            \addplot[name path=baseM, domain=\xmin:\tauHAtoA, draw=none] {\ymin};
            \addplot[name path=baseR, domain=\tauHAtoA:\xmax, draw=none] {\ymin};
            \addplot[name path=jointM, domain=\xmin:\tauHAtoA, draw=none] {\JointLossFunc};
            \addplot[name path=aiR, domain=\tauHAtoA:\xmax, draw=none] {\AILossFunc};
            \addplot[jointshade] fill between[of=baseM and jointM];
            \addplot[aishade] fill between[of=baseR and aiR];
          \fi
          \ifx\LP@regime@tmp\LP@reg@high
            \addplot[name path=baseL, domain=\xmin:\tauHtoA, draw=none] {\ymin};
            \addplot[name path=baseR, domain=\tauHtoA:\xmax, draw=none] {\ymin};
            \addplot[name path=humanL, domain=\xmin:\tauHtoA, draw=none] {\HumanLossFunc};
            \addplot[name path=aiR, domain=\tauHtoA:\xmax, draw=none] {\AILossFunc};
            \addplot[humanshade] fill between[of=baseL and humanL];
            \addplot[aishade] fill between[of=baseR and aiR];
          \fi
        \else
          \ifx\LP@regime@tmp\LP@reg@intermediate
            \addplot[humanshade] coordinates {(\xmin,\ymin) (\tauHtoHA,\ymin) (\tauHtoHA,\ymax) (\xmin,\ymax)} \closedcycle;
            \addplot[jointshade] coordinates {(\tauHtoHA,\ymin) (\tauHAtoA,\ymin) (\tauHAtoA,\ymax) (\tauHtoHA,\ymax)} \closedcycle;
            \addplot[aishade] coordinates {(\tauHAtoA,\ymin) (\xmax,\ymin) (\xmax,\ymax) (\tauHAtoA,\ymax)} \closedcycle;
          \fi
          \ifx\LP@regime@tmp\LP@reg@low
            \addplot[jointshade] coordinates {(\xmin,\ymin) (\tauHAtoA,\ymin) (\tauHAtoA,\ymax) (\xmin,\ymax)} \closedcycle;
            \addplot[aishade] coordinates {(\tauHAtoA,\ymin) (\xmax,\ymin) (\xmax,\ymax) (\tauHAtoA,\ymax)} \closedcycle;
          \fi
          \ifx\LP@regime@tmp\LP@reg@high
            \addplot[humanshade] coordinates {(\xmin,\ymin) (\tauHtoA,\ymin) (\tauHtoA,\ymax) (\xmin,\ymax)} \closedcycle;
            \addplot[aishade] coordinates {(\tauHtoA,\ymin) (\xmax,\ymin) (\xmax,\ymax) (\tauHtoA,\ymax)} \closedcycle;
          \fi
        \fi
    \fi
    \edef\LP@opt@tmp{\LP@optimalIndicator}%
    \ifx\LP@opt@tmp\LP@opt@bold
      \addplot[humanloss, line width=\ThinCurveWidth] {\HumanLossFunc};
      \ifLP@showLegend \addlegendentry{$L_H$}\fi
      \addplot[ailoss, line width=\ThinCurveWidth] {\AILossFunc};
      \ifLP@showLegend \addlegendentry{$L_A$}\fi
      \addplot[jointloss, line width=\ThinCurveWidth] {\JointLossFunc};
      \ifLP@showLegend \addlegendentry{$L_{HA}^{CN}$}\fi
      \ifx\LP@regime@tmp\LP@reg@intermediate
        \addplot[humanloss, line width=\ThickCurveWidth, domain=\xmin:\tauHtoHA, forget plot] {\HumanLossFunc};
        \addplot[jointloss, line width=\ThickCurveWidth, domain=\tauHtoHA:\tauHAtoA, forget plot] {\JointLossFunc};
        \addplot[ailoss, line width=\ThickCurveWidth, domain=\tauHAtoA:\xmax, forget plot] {\AILossFunc};
      \fi
      \ifx\LP@regime@tmp\LP@reg@low
        \addplot[jointloss, line width=\ThickCurveWidth, domain=\xmin:\tauHAtoA, forget plot] {\JointLossFunc};
        \addplot[ailoss, line width=\ThickCurveWidth, domain=\tauHAtoA:\xmax, forget plot] {\AILossFunc};
      \fi
      \ifx\LP@regime@tmp\LP@reg@high
        \addplot[humanloss, line width=\ThickCurveWidth, domain=\xmin:\tauHtoA, forget plot] {\HumanLossFunc};
        \addplot[ailoss, line width=\ThickCurveWidth, domain=\tauHtoA:\xmax, forget plot] {\AILossFunc};
      \fi
    \else
        \addplot[humanloss] {\HumanLossFunc};
        \ifLP@showLegend \addlegendentry{$L_H$}\fi
        \addplot[ailoss] {\AILossFunc};
        \ifLP@showLegend \addlegendentry{$L_A$}\fi
        \addplot[jointloss] {\JointLossFunc};
        \ifLP@showLegend \addlegendentry{$L_{HA}^{CN}$}\fi
    \fi
    \ifLP@showCurveLabels
      \edef\LP@lossLabelLocTmp{\LP@lossLabelLocation}%
      \ifx\LP@lossLabelLocTmp\LP@lossLabelLoc@right
        \node[loss curve labels, \humancolor, anchor=south east] at (axis cs:\xLossLabelRight,\yHumanRight) {\HumanCurveLabel};
        \node[loss curve labels, \jointcolor, anchor=south east, rotate=-15] at (axis cs:\xLossLabelRight,\yJointRight) {\JointCurveLabel};
        \node[loss curve labels, \aicolor,   anchor=south east, rotate=-35] at (axis cs:\xLossLabelRight,\yAIRight)   {\AICurveLabel};
      \else
        \node[loss curve labels, \humancolor] at (axis cs:\xWithhold,\yHumanMid) {\HumanCurveLabel};
        \node[loss curve labels, \jointcolor, rotate=-15] at (axis cs:\xAugment,\yJointMid) {\JointCurveLabel};
        \node[loss curve labels, \aicolor, rotate=-40] at (axis cs:\xAutomate,\yAIMid) {\AICurveLabel};
      \fi
    \fi
    \ifx\LP@regime@tmp\LP@reg@intermediate
      \addplot[intersect] coordinates {(\tauHtoHA,\yH)};
      \addplot[intersect] coordinates {(\tauHAtoA,\yAatTauHAtoA)};
      \addplot[axis to point] coordinates {(\tauHtoHA,\ymax) (\tauHtoHA,\ymin)};
      \addplot[axis to point] coordinates {(\tauHAtoA,\ymax) (\tauHAtoA,\ymin)};
    \fi
    \ifx\LP@regime@tmp\LP@reg@low
      \addplot[intersect] coordinates {(\tauHAtoA,\yAatTauHAtoA)};
      \addplot[axis to point] coordinates {(\tauHAtoA,\ymax) (\tauHAtoA,\ymin)};
    \fi
    \ifx\LP@regime@tmp\LP@reg@high
      \addplot[intersect] coordinates {(\tauHtoA,\yH)};
      \addplot[axis to point] coordinates {(\tauHtoA,\ymax) (\tauHtoA,\ymin)};
    \fi
    \ifx\LP@regime@tmp\LP@reg@intermediate
      \node[xaxis regime labels, \humancolor] at (axis cs:\xWithhold,\ymin) {\WithholdActionLabel};
      \node[xaxis regime labels, \jointcolor
      , yshift=-2pt
      ] at (axis cs:\xAugment,\ymin) {\shortstack{\AugmentActionLabelBreak}};
      \node[xaxis regime labels, \aicolor] at (axis cs:\xAutomate,\ymin) {\AutomateActionLabel};
    \fi
    \ifx\LP@regime@tmp\LP@reg@low
      \node[xaxis regime labels, \jointcolor] at (axis cs:\xAugment,\ymin) {\AugmentActionLabel};
      \node[xaxis regime labels, \aicolor] at (axis cs:\xAutomate,\ymin) {\AutomateActionLabel};
    \fi
    \ifx\LP@regime@tmp\LP@reg@high
      \node[xaxis regime labels, \humancolor] at (axis cs:\xWithhold,\ymin) {\WithholdActionLabel};
      \node[xaxis regime labels, \aicolor] at (axis cs:\xAutomate,\ymin) {\AutomateActionLabel};
    \fi
    \end{axis}%
  \end{tikzpicture}%
}%

\begin{figure}
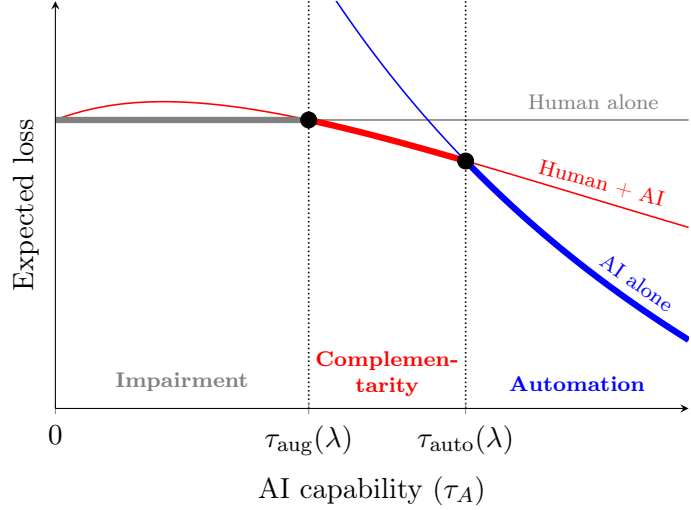
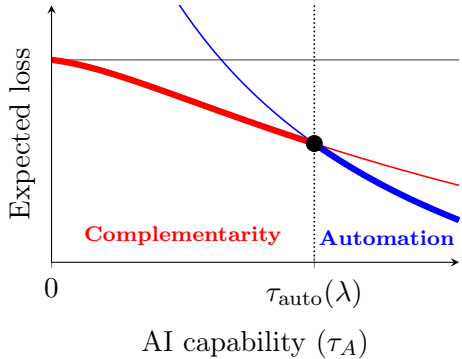
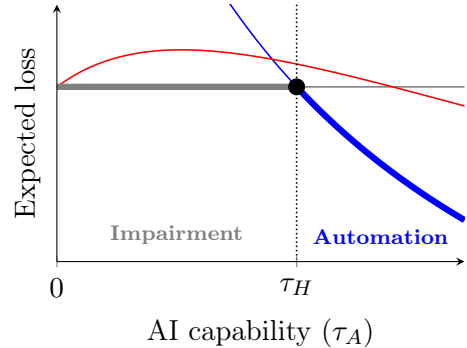

    \centering

    \begin{subfigure}{\linewidth}
        \centering
        \DrawLossPanel{
            lambdaval=\lambdamid,
            xmax=\xmaxLambdaMid,
            ymin=\yminLambdaMid,
            ymax=\ymaxLambdaMid,
            width=10cm,
            height=7cm,
            regime=intermediate,
            showCurveLabels=true,
            showLegend=false,
        }
        \caption{Intermediate overlap $(1/2 < \lambda < \bar{\lambda})$. Here, $\lambda=0.67$.}
        \label{fig:loss-intermediate}
    \end{subfigure}

    \par\medskip

    \begin{subfigure}{0.49\linewidth}
        \centering
        \DrawLossPanel{
            lambdaval=\lambdalow,
            xmax=\xmaxLambdaLow,
            ymin=\yminLambdaLow,
            ymax=\ymaxLambdaLow,
            width=7cm,
            height=5cm,
            regime=low,
            showCurveLabels=false,
            showLegend=false,
        }
        \caption{Low overlap $(\lambda \leq 1/2)$. Here, $\lambda=0.45$.}
        \label{fig:loss-low}
    \end{subfigure}
    \hfill
    \begin{subfigure}{0.49\linewidth}
        \centering
        \DrawLossPanel{
            lambdaval=\lambdahigh,
            xmax=\xmaxLambdaHigh,
            ymin=\yminLambdaHigh,
            ymax=\ymaxLambdaHigh,
            width=7cm,
            height=5cm,
            regime=high,
            showCurveLabels=false,
            showLegend=false,
        }
        \caption{High overlap $(\lambda \geq \bar{\lambda})$. Here, $\lambda=0.85$.}
        \label{fig:loss-high}
    \end{subfigure}

    \caption{Expected loss as a function of AI capability $\tau_A$ in the conceptual model (Section~\ref{sec:conceptual}). Shown are losses of human acting alone ($L_{\mathrm{H}}$, gray), AI acting alone ($L_{\mathrm{AI}}$, blue), and AI-assisted human who exhibits correlation neglect ($\widehat L_{\mathrm{H-AI}}$, red). Regime thresholds and the best-performing agent in each regime are labeled on the $x$-axis, with the smallest lost curve thickened. Each panel fixes $\tau_Y=1, \tau_H=1$ and $\lambda$.}
    \label{fig:loss-three-panels}
\end{figure}

\begin{itemize}
\item \textbf{Low overlap ($0<\lambda<\frac12$, Figure~\ref{fig:loss-low}).}\\
In these domains, overlap is mild, so despite human behavioral misspecification---correlation neglect---the human always benefits from AI assistance. Improving AI capabilities transitions through \textit{two regimes}.
\begin{align*}
&\text{Regime I (\textit{Complementarity}): } \tau_A\in(0,\tau_{\mathrm{auto}}(\lambda))
&&\Rightarrow
\widehat L_{\mathrm{H-AI}}=\min\{L_{\mathrm{H}},\widehat L_{\mathrm{H-AI}},L_{\mathrm{AI}}\},\\
&\text{Regime II (\textit{Automation}): } \tau_A\in(\tau_{\mathrm{auto}}(\lambda),+\infty)
&&\Rightarrow
L_{\mathrm{AI}}=\min\{L_{\mathrm{H}},\widehat L_{\mathrm{H-AI}}, L_{\mathrm{AI}}\}.
\end{align*}
Weaker AI benefits the human but is not strong enough to automate the tasks. As AI capabilities grow, human's behavioral misspecification becomes penalizing and automation eventually dominates.

\item \textbf{Intermediate overlap ($\tfrac12<\lambda<\bar\lambda$, Figure~\ref{fig:loss-intermediate}).}\\
Overlap here is large enough that weak AI can penalize the human where the behavioral misspecification dominates the AI marginal value. The overlap is however not large enough to render complementarity impossible. In these domains, augmentation happens before automation and the thresholds satisfy
\[
0<\tau_{\mathrm{aug}}(\lambda)<\tau_{\mathrm{auto}}(\lambda).
\]
\textit{Three regimes} arise when increasing AI capabilities.
\begin{align*}
&\text{Regime I (\textit{Impairment}): } \tau_A\in(0,\tau_{\mathrm{aug}}(\lambda))
&&\Rightarrow
L_{\mathrm{H}}=\min\{L_{\mathrm{H}},\widehat L_{\mathrm{H-AI}},L_{\mathrm{AI}}\},\\
&\text{Regime II (\textit{Complementarity}): } \tau_A\in(\tau_{\mathrm{aug}}(\lambda),\tau_{\mathrm{auto}}(\lambda))
&&\Rightarrow
\widehat L_{\mathrm{H-AI}}=\min\{L_{\mathrm{H}},\widehat L_{\mathrm{H-AI}},L_{\mathrm{AI}}\},\\
&\text{Regime III (\textit{Automation}): } \tau_A\in(\tau_{\mathrm{auto}}(\lambda),+\infty)
&&\Rightarrow
L_{\mathrm{AI}}=\min\{L_{\mathrm{H}},\widehat L_{\mathrm{H-AI}},L_{\mathrm{AI}}\}.
\end{align*}
In this case, complementarity only arises in the intermediate region:
only after AI is capable enough so that its marginal value overcomes human's behavioral misspecification, but
before AI becomes so strong that it becomes optimal to bypass the human.

\item \textbf{High overlap ($\bar\lambda< \lambda<1$, Figure~\ref{fig:loss-high}).}
In these domains, the behavior misspecification's impact is so severe, that the AI-assisted human is always dominated by the human-alone or the AI-alone: the human-AI interaction is significantly penalizing.

Here, we transition directly from human-only to AI-only regimes.
\begin{align*}
&\text{Regime I (\textit{Impairment}): } \tau_A\in(0,\tau_H)
&&\Rightarrow
L_{\mathrm{H}}=\min\{L_{\mathrm{H}},\widehat L_{\mathrm{H-AI}},L_{\mathrm{AI}}\},\\
&\text{Regime II (\textit{Automation}): } \tau_A\in(\tau_H,+\infty)
&&\Rightarrow
L_{\mathrm{AI}}=\min\{L_{\mathrm{H}},\widehat L_{\mathrm{H-AI}},L_{\mathrm{AI}}\}.
\end{align*}
\end{itemize}

\def\tauzero{1}
\def\tauH{1}

\def\phasexmin{0}
\def\phasexminsafe{0.02}  %
\def\phasexmax{2.2}
\def\phaseymin{0}
\def\phaseymax{1}

\def\numsamples{300}

\pgfplotsset{
    phase hard boundary/.style={line width=\MediumCurveWidth, black, solid},
    phase soft boundary/.style={line width=\ThinCurveWidth, darkgray, dashed},
    phase feasible boundary/.style={line width=\ThinCurveWidth, black, solid},
    phase ovlp1 boundary/.style={line width=\ThinCurveWidth, black, solid},
}

\tikzset{
    phase regime labels/.style={
        anchor=center, font=\normalsize\bfseries, inner sep=3pt,
    },
    phase label human/.style={
        phase regime labels,
        darkgray  %
    },
    phase label joint/.style={
        phase regime labels,
        \jointcolor,
        fill=\jointshadedarkcolor,
        fill opacity=1,
        text opacity=1,
        draw=none,
        rounded corners=1pt,
    },
    phase label ai/.style={
        phase regime labels,
        \aicolor,
        rotate=-20
    },
    boundary labels/.style={
        anchor=north, font=\small,
        inner sep=4pt
    }
}

\pgfmathsetmacro{\phaselambdastar}{(\tauzero+2*\tauH)/(2*(\tauzero+\tauH))} %

\pgfmathsetmacro{\Tsum}{\tauzero+\tauH}
\pgfmathsetmacro{\xHvsJcap}{(-\Tsum + sqrt(\Tsum*\Tsum + 8*\tauH*\Tsum))/2}

\pgfmathsetmacro{\xWithhold}{0.5*(\phasexmin+\tauH)}
\pgfmathsetmacro{\xAugment}{\tauH}
\pgfmathsetmacro{\xAutomate}{0.5*(\tauH+\phasexmax)}

\pgfmathsetmacro{\lamHJwith}{0.5*(1 + \xWithhold/(\tauzero+\tauH))}
\pgfmathsetmacro{\lamHJaug}{0.5*(1 + \xAugment/(\tauzero+\tauH))}

\pgfmathsetmacro{\lamJAaug}{\tauH*(\xAugment+\tauzero+\tauH)/(2*\xAugment*(\xAugment+\tauzero))}
\pgfmathsetmacro{\lamJAauto}{\tauH*(\xAutomate+\tauzero+\tauH)/(2*\xAutomate*(\xAutomate+\tauzero))}

\pgfmathsetmacro{\lamFeasauto}{min(\phaseymax, \tauH/\xAutomate)}

\pgfmathsetmacro{\yWithhold}{0.5*(\phaseymax + \lamHJwith)}

\pgfmathsetmacro{\yAugment}{ifthenelse(\xAugment < \tauH,
    0.5*(\phaseymin + \lamHJaug),
    0.5*(\phaseymin + \lamJAaug)
)}

\pgfmathsetmacro{\yAutomate}{0.5*(\lamJAauto + \lamFeasauto)}

\pgfmathsetmacro{\yWithhold}{max(\phaseymin, min(\phaseymax, \yWithhold))}
\pgfmathsetmacro{\yAugment}{max(\phaseymin, min(\phaseymax, \yAugment))}
\pgfmathsetmacro{\yAutomate}{max(\phaseymin, min(\phaseymax, \yAutomate))}

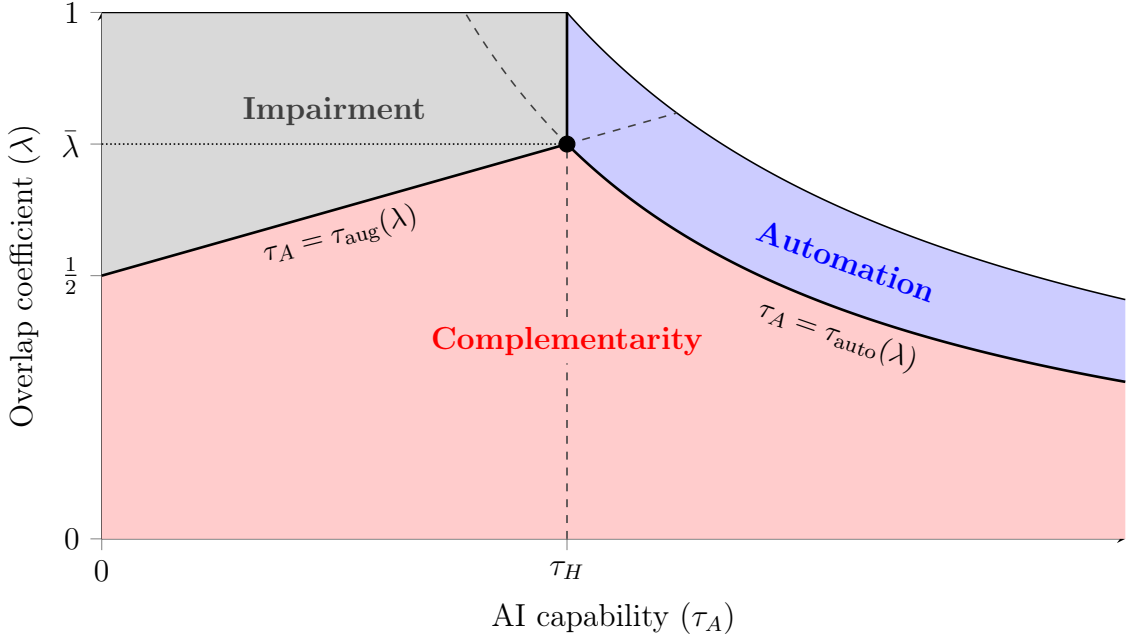
\begin{figure}[!t]
    \centering

    \begin{tikzpicture}
    \begin{axis}[
        width=0.92\linewidth,
        height=0.52\linewidth,
        xmin=\phasexmin, xmax=\phasexmax,
        ymin=\phaseymin, ymax=\phaseymax,
        xlabel={\AbilityAxisLabel},
        ylabel={\OverlapAxisLabel},
        axis lines=left,
        tick align=outside,
        clip=true,
        legend style={draw=none, fill=none, at={(0.02,0.98)},       
            anchor=north west},
        unbounded coords=jump,
        restrict y to domain=\phaseymin:\phaseymax,
        xtick={\phasexmin, \tauH},
        xticklabels={$0$, $\tau_H$},
        ytick={\phaseymin, 0.5, \phaselambdastar, 1},
        yticklabels={$0$, $\frac{1}{2}$, $\bar{\lambda}$, $1$},
    ]

    \addplot[name path=feas_left, draw=none, domain=\phasexmin:\tauH, samples=\numsamples]
        ({x},{1});
    \addplot[name path=feas_right, draw=none, domain=\tauH:\phasexmax, samples=\numsamples]
        ({x},{\tauH/x});
    
    \addplot[name path=HvsA, draw=none, line width=\MediumCurveWidth, black] coordinates {(\tauH,\phaseymin) (\tauH,\phaseymax)};
    
    \addplot[name path=HvsJ, draw=none, domain=\phasexmin:\phasexmax, samples=\numsamples]
        ({x},{0.5*(1 + x/(\tauzero+\tauH))});
    
    \addplot[name path=JvsA, draw=none, domain=\phasexminsafe:\phasexmax, samples=\numsamples]
        ({x},{\tauH*(x+\tauzero+\tauH)/(2*x*(x+\tauzero))});
    
    \addplot[name path=bottom, draw=none, domain=\phasexmin:\phasexmax, samples=\numsamples]
        ({x},{\phaseymin});

    \addplot[jointshade dark] fill between[of=bottom and HvsJ, soft clip={domain=\phasexmin:\tauH}];
    
    \addplot[humanshade dark] fill between[of=HvsJ and feas_left, soft clip={domain=\phasexmin:\tauH}];
    
    \addplot[jointshade dark] fill between[of=bottom and JvsA, soft clip={domain=\tauH:\phasexmax}];
    
    \addplot[aishade dark] fill between[of=JvsA and feas_right, soft clip={domain=\tauH:\phasexmax}];

    \addplot[phase ovlp1 boundary, black] coordinates {(\phasexmin,1) (\tauH,1)};
    \addplot[phase feasible boundary, black, domain=\tauH:\phasexmax, samples=\numsamples]
        ({x},{\tauH/x});
        
    \addplot[phase soft boundary] coordinates {(\tauH,\phaseymin) (\tauH,\phaselambdastar)};
    \addplot[phase hard boundary] coordinates {(\tauH,\phaselambdastar) (\tauH,1)};
    
    \addplot[phase hard boundary, domain=\phasexmin:\tauH, samples=\numsamples]
        ({x},{0.5*(1 + x/(\tauzero+\tauH))});
    \addplot[phase soft boundary, domain=\tauH:\xHvsJcap, samples=\numsamples]
        ({x},{0.5*(1 + x/(\tauzero+\tauH))});

    \addplot[phase soft boundary, domain=\phasexminsafe:\tauH, samples=\numsamples]
        ({x},{\tauH*(x+\tauzero+\tauH)/(2*x*(x+\tauzero))});
    \addplot[phase hard boundary, domain=\tauH:\phasexmax, samples=\numsamples]
        ({x},{\tauH*(x+\tauzero+\tauH)/(2*x*(x+\tauzero))});

    \addplot[intersect] coordinates {(\tauH,\phaselambdastar)};
    \addplot[axis to point] coordinates {(\phasexmin,\phaselambdastar) (\tauH,\phaselambdastar)};

    \node[phase label human] at (axis cs:\xWithhold,\yWithhold) {\WithholdActionLabel};
    \node[phase label joint] at (axis cs:\xAugment,\yAugment) {\AugmentActionLabel};
    \node[phase label ai] at (axis cs:\xAutomate,\yAutomate) 
    {\AutomateActionLabel};
    
    \node[boundary labels, rotate=14] at (axis cs:\xWithhold,\lamHJwith) {$\tau_A = \HtoHAlabel$};
    \node[boundary labels, rotate=-19] at (axis cs:\xAutomate,\lamJAauto) {$\tau_A = \HAtoAlabel$};

    \end{axis}
    \end{tikzpicture}

    \caption{Phase diagram of agent-optimal regimes in the conceptual model (Section~\ref{sec:conceptual}) when the human exhibits correlation neglect. The diagram varies AI capability $\tau_A$ ($x$-axis) and overlap coefficient $\lambda$ ($y$-axis): each $(\tau_A, \lambda)$ point corresponds to a unique joint distribution of signals $(Y,H,A)$. Regimes are labeled according to the best-performing agent: AI-assisted human (complementarity, red), human alone (impairment, gray), and AI alone (automation, blue). Threshold functions where two agents perform identically are shown as solid curves where they define regime boundaries, and dashed curves elsewhere. Each panel in Figure~\ref{fig:loss-three-panels} shows loss functions along a horizontal slice of this diagram. The figure fixes $\tau_Y=1$ and $\tau_H=1$.
    }
    \label{fig:comp-phases}
\end{figure}

A phase diagram of all three complementarity regimes, and how they vary with overlap $\lambda$ and AI capability $\tau_A$ when other model parameters are fixed, is shown in Figure~\ref{fig:comp-phases}.

\section{Conclusion and Discussion}

This paper adopts a decision-theoretic framework to study how AI assistance affects human decision-making, especially when humans act imperfectly, such as by exhibiting correlation neglect. We introduce a Bayesian model in which an AI-assisted human combines informative signals from herself and the AI, and decompose the benefit of AI assistance into the \textit{marginal value} of AI information (relative to the human) and a \textit{behavioral penalty} that arises from human misuse. We then quantify the performance of an AI-assisted human in a micro-founded conceptual model as a function of two key factors: the AI's standalone \textit{capability}, and the \textit{information overlap} that captures redundancy of AI information for the human. This analysis characterizes three distinct regimes of human-AI interaction: \textit{impairment, automation, and complementarity.}
Crucially, our results show that the value of AI assistance does not depend solely on its own capability: other factors, such as information overlap and human behavioral biases, play a central role in designing optimal forms of human--AI collaboration.

Our results carry important design implications. On the AI side, our result suggests that AI systems should be trained to predict what humans miss, not what they already know. Residual learning approaches that target the component of the state unexplained by human information can reduce effective overlap while preserving accuracy. On the human side, initiatives to reduce correlation neglect---such as through training or improved user interfaces---may expand the range of AI models and capabilities under which complementarity can be achieved.

Several directions for future work emerge from our analysis. First, the framework can be extended beyond correlation neglect to other behavioral biases, such as overconfidence in one's own expertise or algorithm aversion. Another natural extension is to model repeated human--AI interaction, in which the human gradually learns the conditional dependence between signals and reduce misspecification in her beliefs over time. Finally, empirical measurement of the overlap coefficient and its effects on human--AI performance across different domains, using data on medical diagnosis or human--AI interaction design, can ground our theoretical predictions and inform practical deployment decisions.

\section*{Acknowledgment}
We are grateful to the MIT
Gen-AI Consortium for financial support.

\bibliographystyle{plainnat}
\bibliography{dingwenkbib,human-ai}
\newpage %
\begin{appendices}
\section{Heterogeneous-precision extension (including Proof of Proposition~\ref{prop:lambda_concentration})}
\label{appendix:hetero}
Here, we provide an extension where information cues can be heterogeneous. We prove a generalized version of Proposition~\ref{prop:lambda_concentration}, and Proposition~\ref{prop:lambda_concentration} itself directly follows as a special case.

In this extension, let $\cS=\{s_1,\dots,s_N\}$ be a collection of conditionally independent \emph{heterogeneous} Gaussian primitive cues
about state $Y$:
\[
s_i \perp\!\!\!\perp s_j \mid Y \quad (i\neq j),
\qquad
s_i \mid Y \sim \mathcal{N} \left(Y, N\tau_i^{-1}\right),
\]
where $\tau_i>0$ is cue $i$'s normalized precision.

\begin{assumption}[Bounded variation]
\label{assum:bounded}
There exist constants $0<\underline\tau\le \overline\tau<\infty$ such that for all
$i\in S$,
\[
\underline\tau \le \tau_i \le \overline\tau.
\]
\end{assumption}

The human observes a subset $H\subseteq S$ and the AI observes a subset $A\subseteq S$.
In the heterogeneous-precision setup, the natural sufficient statistic for $Y$ in each
subset is the precision-weighted average:
\[
H := \frac{\sum_{i\in H}\tau_i s_i}{\sum_{i\in H}\tau_i},
\qquad
A := \frac{\sum_{i\in A}\tau_i s_i}{\sum_{i\in A}\tau_i}.
\]
Define the total precision masses
\[
T_H := \sum_{i\in H}\tau_i,
\qquad
T_A := \sum_{i\in A}\tau_i.
\]
Then $H\mid Y \sim \mathcal{N} \left(Y, N T_H^{-1}\right)$ and
$A\mid Y \sim \mathcal{N} \left(Y, N T_A^{-1}\right)$.

\begin{lemma}[Weighted overlap formula]
\label{lem:weightedlambda}
Recall the overlap coefficient
\[
\lambda := \frac{\Cov(H,A\mid Y)}{\Var(H\mid Y)}.
\]
Under this extension model,
\[
\lambda
=
\frac{\sum_{i\in A\cap H}\tau_i}{\sum_{i\in A}\tau_i}
=
\frac{T_{A\cap H}}{T_A},
\qquad
\text{where } T_{A\cap H}:=\sum_{i\in A\cap H}\tau_i.
\]
In particular, if $\tau_i=\tau$ for some $\tau>0$ and for all $i$, then $\lambda = |A\cap H|/|A|$.
\end{lemma}

\begin{proof}
Direct computation shows that
\begin{align*}
\Cov(H,A\mid Y)
&= \sum_{i\in A\cap H} \frac{\tau_i}{T_H}\frac{\tau_i}{T_A}\cdot \Var(s_i\mid Y) = \frac{N}{T_H T_A}\sum_{i\in A\cap H}\tau_i
= \frac{N}{T_H T_A}\,T_{A\cap H}.
\end{align*}
Similarly,
\[
\Var(H\mid Y)
= \sum_{i\in H}\frac{\tau_i^2}{T_H^2}\cdot \Var(s_i\mid Y)
= \frac{N}{T_H}.
\]
Therefore,
\[
\lambda
= \frac{\Cov(H,A\mid Y)}{\Var(H\mid Y)}
= \frac{\frac{N}{T_H T_A}T_{A\cap H}}{\frac{N}{T_H}}
= \frac{T_{A\cap H}}{T_A}.
\]
\end{proof}

The AI-accessible pool $\cS_{\mathrm{ai-acc}}\subseteq \cS$ and its complement $\cS_{\mathrm{ai-noacc}}$ is defined similarly.
The AI constructs $A$ by drawing a subset of size $a|\cS|$ uniformly at random from $\cS_{\mathrm{ai-acc}}$.

Define the total precision mass in the AI-accessible pool and in the human and AI-accessible intersection:
\[
T_{\mathrm{ai-acc}} := \sum_{i\in \cS_{\mathrm{ai-acc}}}\tau_i,
\qquad
T_{H\cap\mathrm{ai-acc}} := \sum_{i\in \cH\cap \cS_{\mathrm{ai-acc}}}\tau_i.
\]
Define also the \emph{precision-weighted overlap rate}
\[
\theta_N := \frac{T_{H\cap\mathrm{ai-acc}}}{T_{\mathrm{ai-acc}}}.
\]

\begin{proposition}[Overlap stabilization with heterogeneity]
\label{prop:hetero}
Under Assumption~\ref{assum:bounded}, let $\cA$ be drawn uniformly among all $(a|\cS|)$-element subsets
of $\cS_{\mathrm{ai-acc}}$. Let $\lambda$ be the overlap coefficient defined in
Lemma~\ref{lem:weightedlambda}. Then
\[
\lambda \toP \theta_N
\quad\text{as } N=|\cS|\to\infty.
\]
In particular, if $\theta_N\to \theta$ for some constant $\theta\in[0,1]$, then $\lambda\toP \theta$.
\end{proposition}

\begin{proof}
Write $M:=|\cS_{\mathrm{ai-acc}}|$ and $n_N:=aN$.
For each $i\in \cS_{\mathrm{ai-acc}}$, define two bounded arrays
\[
w_i := \tau_i,
\qquad
u_i := \tau_i\mathbf{1}\{i\in \cH\}.
\]
Then
\[
\sum_{i\in \cA}\tau_i=\sum_{i\in \cA} w_i,
\qquad
\sum_{i\in \cA\cap \cH}\tau_i=\sum_{i\in \cA} u_i,
\qquad\text{so}\qquad
\lambda=\frac{\sum_{i\in \cA}u_i}{\sum_{i\in \cA}w_i}.
\]
Define sample means over the random subset $\cA$:
\[
\widehat{\mu}_{u}:=\frac{1}{n_N}\sum_{i\in \cA}u_i,
\qquad
\widehat{\mu}_{w}:=\frac{1}{n_N}\sum_{i\in \cA}w_i.
\]
Also define the corresponding population means over $\cS_{\mathrm{ai-acc}}$:
\[
\mu_{u}:=\frac{1}{M}\sum_{i\in \cS_{\mathrm{ai-acc}}}u_i
= \frac{T_{H\cap \mathrm{ai-acc}}}{M},
\qquad
\mu_{w}:=\frac{1}{M}\sum_{i\in \cS_{\mathrm{ai-acc}}}w_i
= \frac{T_{\mathrm{ai-acc}}}{M}.
\]
Then
\[
\lambda=\frac{\widehat{\mu}_{u}}{\widehat{\mu}_{w}},
\qquad
\theta_N=\frac{\mu_{u}}{\mu_{w}}.
\]
Because $\cA$ is drawn uniformly among all $n_N$-subsets of $\cS_{\mathrm{ai-acc}}$, it is a simple random
sample without replacement from the finite population $\cS_{\mathrm{ai-acc}}$. Therefore, for $x\in\{u,w\}$, (see, e.g., \cite{lohr2021sampling})
\[
\mathbb{E}[\widehat{\mu}_{x}] = \mu_{x},
\qquad
\mathrm{Var}(\widehat{\mu}_{x})
=
\Bigl(1-\frac{n_N}{M}\Bigr)\frac{S_{x}^2}{n_N},
\]
where $S_{x}^2=\frac{1}{M-1}\sum_{i\in \cS_{\mathrm{ai-acc}}}(x_i-\bar x)^2$ denotes the finite-population variance of $\{x_i\}_{i\in \cS_{\mathrm{ai-acc}}}$.

Under Assumption~\ref{assum:bounded}, we have $w_i\in[\underline{\tau},\overline{\tau}]$ and $u_i\in[0,\overline{\tau}]$.
Therefore,
\[
S_{w}^2 \leq (\overline{\tau}-\underline{\tau})^2,
\qquad
S_{u}^2 \le \overline{\tau}^2,
\]
and thus
\[
\Var(\widehat{\mu}_{w})\leq \frac{(\overline{\tau}-\underline{\tau})^2}{n_N},
\qquad
\Var(\widehat{\mu}_{u})\leq \frac{\overline{\tau}^2}{n_N}.
\]
Since $n_N=aN\to\infty$, both variances go to $0$. By Chebyshev's inequality, for every $\varepsilon>0$,
\[
\mathbb{P}\bigl(|\widehat{\mu}_{w}-\mu_{w}|>\varepsilon\bigr)\to 0,
\qquad
\mathbb{P}\bigl(|\widehat{\mu}_{u}-\mu_{u}|>\varepsilon\bigr)\to 0,
\]
i.e.,
\[
\widehat{\mu}_{w}-\mu_{w}\xrightarrow{p}0,
\qquad
\widehat{\mu}_{u}-\mu_{u}\xrightarrow{p}0.
\]

Now let's show the convergence of the ratio.
Note that $\widehat{\mu}_{w}\ge \underline{\tau}>0$ almost surely, and likewise $\mu_{w}\ge \underline{\tau}>0$.
Using
\[
\lambda-\theta_N
=
\frac{\widehat{\mu}_{u}}{\widehat{\mu}_{w}}-\frac{\mu_{u}}{\mu_{w}}
=
\frac{\widehat{\mu}_{u}-\mu_{u}}{\widehat{\mu}_{w}}
-\frac{\mu_{u}}{\mu_{w}}\cdot
\frac{\widehat{\mu}_{w}-\mu_{w}}{\widehat{\mu}_{w}},
\]
we obtain
\[
|\lambda-\theta_N|
\le
\frac{1}{\underline{\tau}}\,|\widehat{\mu}_{u}-\mu_{u}|
+
\frac{\overline{\tau}}{\underline{\tau}^2}\,|\widehat{\mu}_{w}-\mu_{w}|.
\]
The right-hand side converges to $0$ in probability; therefore,
\[
\lambda\xrightarrow{p}\theta_N.
\]
\end{proof}

\begin{remark}[Relationship to the baseline model]
Our baseline model corresponds to the case
\[
\tau_i=\tau \ \forall i, \quad |\cS_{\mathrm{ai-acc}}|=mN, \quad |\cH\cap \cS_{\mathrm{ai-acc}}|=kN,
\]
so that
\[\theta_N=\frac{\sum_{i\in \cH\cap \cS_{\mathrm{ai-acc}}}\tau_i}{\sum_{i\in  \cS_{\mathrm{ai-acc}}}\tau_i}=\frac{k}{m},\] and Proposition~\ref{prop:hetero} reduces to the original overlap stabilization
$\lambda\toP k/m$.
\end{remark}

\section{Missing Proofs}
\subsection{Proof of Proposition~\ref{prop:relative_complement}}
Consider a binary state $Y\in\{0,1\}$ with prior $Y\sim\bern(\frac12)$, real-valued decision $d\in\RR$, and  logistic loss $L(Y,d)=-Y\log d-(1-Y)\log (1-d)$ where log base is always set to 2.

First, we show that the ratio $\frac{v (A\mid H)}{v(A)}$ can take any value in $[1,\infty)$. Take two independent noise bits $U\sim \bern (p)$ and $V\sim \bern (q)$. Let
\[
A=Y\oplus U \oplus V,
\]
and 
\[
H=V.
\]
Then we have that
\begin{align*}
v(A\mid H)&=I(Y;A\mid H)\\
& = I (Y;Y\oplus U \oplus V\mid V)\\
&= I (Y;Y\oplus U)\\
&= 1-h (p),
\end{align*}
where $h(\cdot)$ is the binary entropy function.

For $v(A)$, we have 
\begin{align*}
v(A)&=I(Y;A)\\
& = I (Y;Y\oplus (U \oplus V))\\
&= 1-h (p+q-2pq).
\end{align*}
Therefore, the ratio is
\[
\frac{v(A\mid H)}{v(A)}=\frac{1-h (p)}{1-h (p+q-2pq)}.
\]
When $q=0$ we have $\frac{v(A\mid H)}{v(A)}=1$. When $q\uparrow \frac12$, we have $\frac{v(A\mid H)}{v(A)}\rightarrow +\infty$. By continuity, for any $t\in[1,\infty)$ there exists some $q$ such that the ratio equals $t$.

Second, we show that the ratio $\frac{v (A\mid H)}{v(A)}$ can take any value in $[0,1)$. Take an independent noise bit $U\sim \bern(p)$ and an independent indicator $M\sim \bern (1-t)$. Let 
\[
A= Y\oplus U,
\]
and
\[
H=
\begin{cases}
    A, & \text{if $M=1$}\\
    \perp, & \text{if $M=0$},
\end{cases}
\]
where $\perp$ indicates a null signal with no information. 

In this case, 
\[
v(A)=I(Y;A)=1- h(p),
\]
while
\[
v(A\mid H)= I(Y;A\mid H) = \Pr(H=\ \perp) I(Y;A)=t(1-h(p)),
\]
so that
\[
\frac{v(A\mid H)}{v(A)}=t.
\]

\subsection{Proof of Theorem~\ref{thm: complementarity gap}}

\newcommand{\anglebkt}[1]{\left\langle#1\right\rangle}
\newcommand{\sqbkt}[1]{\left[#1\right]}
\newcommand{\parbkt}[1]{\left(#1\right)}
\newcommand{\curbkt}[1]{\left\{#1\right\}}
\newcommand{\Norm}[1]{\left\lVert#1\right\rVert}
\newcommand{\R}{\mathbb{R}}
\newcommand{\E}{\mathbb{E}}
\newcommand{\dopt}{\delta^{\star}}
\newcommand{\dbelief}{\widehat{\delta}}

\newcommand{\cond}[2]{\left.#1\,\middle|\,#2\right.}
\newcommand{\condp}[2]{\left(#1\,\middle|\,#2\right)}
\newcommand{\condsq}[2]{\left[#1\,\middle|\,#2\right]}
\newcommand{\condcur}[2]{\left\{#1\,\middle|\,#2\right\}}

\newcommand{\supp}{\mathrm{supp}}

We first recall the definition of Bregman loss functions (BLFs), as per \cite{Banerjee2005Bregman}:
\begin{definition}[Bregman Loss Functions]
    Let $\phi: \mathbb{R}^n \mapsto \mathbb{R}$ be a strictly convex and differentiable function. Then, the Bregman loss function $D_\phi: \mathbb{R}^n \times \mathbb{R}^n \mapsto \mathbb{R}$ is defined as 
    \begin{align*}
        D_\phi (y,d) &= \phi (y) - \phi (d) - \anglebkt{ y-d, \nabla \phi (d) }.
    \end{align*}
\end{definition}
Note that $D_\phi$ is asymmetric.
Two notable examples of BLFs are:
\begin{itemize}
    \item Squared loss: Let $\phi (x) := \anglebkt{x,x}= \Norm{x}^2$. Then, the BLF is the $L_2$-loss: $D_\phi (y,d) = \Norm{y-d}^2$. In particular, when $n=1$ and $Y$ is a scalar, the BLF is the mean squared error (MSE): $D_\phi \parbkt{y,d} = \parbkt{y-d}^2$.
    \item Cross-entropy loss: Let $y$ be deterministic vector in the $n$-simplex, and $\phi \parbkt{x} := \sum_{i=1}^n x_i \log x_i$ be the negative entropy function. Then, the BLF is the KL divergence: $D_\phi \parbkt{y,d} = \KL \parbkt{y \, \middle\Vert \, d}$. In particular, when $y$ is a one-hot vector, $D_\phi \parbkt{y,d} = - \sum_{i=1}^n \mathbf{1}\sqbkt{y_i=1} \log d_i$ is exactly the log loss, and its expectation
    $$ \E_Y \sqbkt{ D_\phi \parbkt{Y,d} } = H \parbkt{p_Y, d},  \quad \forall d\in \R^n $$
    is the cross-entropy loss, where $p_Y \in \R^n$ gives the distribution of $Y$.
\end{itemize}

Given any information set $S$ (such as $H$, $A$ or $(H,A)$), a Bayesian-optimal agent who observes $S=s$ makes the risk-minimizing decision
\begin{align}
    \dopt \parbkt{s} &:= \argmin_{d\in \R^n} \E \condsq{L\parbkt{Y, d}}{S=s}.  \label{eqn:def_opt_predictor}
\end{align}

We cite two useful properties of BLFs from \cite{Banerjee2005Bregman}:
\begin{lemma}[Conditional Expectation Minimizes BLF]  \label{lemma:cond_expect_optimal_blf}
    Let $Y\in \R^n$ be a random variable, $S$ be an information set, $\dopt \parbkt{S}$ be the optimal predictor given $S$ as defined in~\eqref{eqn:def_opt_predictor}, and $L:=D_\phi$ be a BLF. Then, $\dopt \parbkt{S}$ is uniquely given by the conditional expectation:
    \begin{align}
        \dopt \parbkt{s} &= \E \condsq{Y}{S=s},  \quad \forall s\in \supp \parbkt{S}.  \label{eqn:cond_expect_optimal_blf}
    \end{align}
\end{lemma}
\begin{lemma}[Three-Point Lemma]  \label{lemma:triangle_blf}
    Let $Y\in \R^n$ be a random variable, $S$ be an information set, $\dopt \parbkt{S} = \E \condsq{Y}{S=s}$, $\dbelief\parbkt{S}$ be an arbitrary predictor given $S$, and $D_\phi$ be a BLF. Then,
    \begin{align}
        \E_{Y,S} \sqbkt{ D_\phi \parbkt{Y, \dbelief\parbkt{S}} } - \E_{Y,S} \sqbkt{ D_\phi \parbkt{Y, \dopt\parbkt{S}} } &= \E_S \sqbkt{ D_\phi \parbkt{\dopt\parbkt{S}, \dbelief\parbkt{S}} }.  \label{eqn:triangle_blf}
    \end{align}
\end{lemma}
Lemma~\ref{lemma:triangle_blf} and Lemma~\ref{lemma:cond_expect_optimal_blf} are Equation (1) and Theorem 1 in~\cite{Banerjee2005Bregman}, respectively. Note that Lemma~\ref{lemma:triangle_blf} immediately implies uniqueness in Lemma~\ref{lemma:cond_expect_optimal_blf}.
\citeauthor{Banerjee2005Bregman} also show that under mild conditions, \eqref{eqn:cond_expect_optimal_blf} holds \textit{only} if the loss function is a BLF, thus establishing BLF as necessary and sufficient.

We now prove Theorem~\ref{thm: complementarity gap}:

\newcommand{\Lhat}{\widehat{L}_{\mathrm{H-AI}}}
\newcommand{\Lopt}{L^{\star}_{\mathrm{H-AI}}}
\newcommand{\LH}{L_{\mathrm{H}}}
\newcommand{\decruleH}{\delta(H)}

\begin{proof}
    Denote expected losses with
    \begin{align*}
        \Lhat &:= \E_{Y,H,A} \sqbkt{ L \parbkt{\dbelief\parbkt{H,A} , Y} },  && \text{(realized AI-assisted human loss)}  \\
        \Lopt &:= \E_{Y,H,A} \sqbkt{ L \parbkt{\dopt\parbkt{H,A} , Y} },  && \text{(optimal AI-assisted human loss)}  \\
        \LH &:= \E_{Y,H,A} \sqbkt{ L \parbkt{\decruleH , Y} }.  && \text{(optimal human-only loss)}
    \end{align*}
    Note that we had defined $\decruleH$ as the Bayesian-optimal human-only decision rule.
    Applying Lemma~\ref{lemma:triangle_blf} to $S=\parbkt{H,A}$, $\dbelief\parbkt{S} = \dbelief\parbkt{H,A}$ and $D_{\phi}=L$ gives
    \begin{align}
        \Lhat - \Lopt &= \E_{H,A} \sqbkt{ L \parbkt{\dopt \parbkt{H,A}, \dbelief\parbkt{H,A}} }.  \label{eqn:thm1_proof_loss_diff_deltahat}
    \end{align}
    Likewise, apply Lemma~\ref{lemma:triangle_blf} to $S=\parbkt{H,A}$ and $\dbelief\parbkt{S} = \decruleH = \E\condsq{Y}{H}$ (the last equality is given by Lemma~\ref{lemma:cond_expect_optimal_blf}), and relate to Definition~\ref{def:marg_info} that gives $v\condp{A}{H}$. This gives
    \begin{align}
        v\condp{A}{H} := \LH - \Lopt &= \E_{H,A} \sqbkt{ L \parbkt{\dopt \parbkt{H,A}, \decruleH} },  \label{eqn:thm1_proof_loss_diff_doptH}
    \end{align}
    where the first equality is exactly the definition of marginal value. From~\eqref{eqn:thm1_proof_loss_diff_deltahat} and~\eqref{eqn:thm1_proof_loss_diff_doptH},
    \begin{align*}
        \LH - \Lhat &= \E_{H,A} \sqbkt{ L \parbkt{\dopt \parbkt{H,A}, \decruleH} } - \E_{H,A} \sqbkt{ L \parbkt{\dopt \parbkt{H,A}, \dbelief\parbkt{H,A}} }  \\
        &= v\condp{A}{H} - \E_{H,A} \sqbkt{ L \parbkt{\dopt \parbkt{H,A}, \dbelief\parbkt{H,A}} },
    \end{align*}
    which completes the proof.
\end{proof}

\subsection{Proof of Lemma~\ref{lem:decomposition}}
We first prove that $\tilde A\perp\!\!\!\perp H\mid Y$. Note that
\[
\Cov(\tilde A,H\mid Y)=\frac{1}{1-\lambda}\left(\Cov(A,H)-\lambda\cdot \Var(H\mid Y)\right)=0.
\]
Since $(H,\tilde A)$ is also jointly Gaussian distributed conditional on $Y$, we have $\tilde A\perp\!\!\!\perp H\mid Y$.

Therefore, from conditional independence, we obtain
\[
\Var (A\mid Y)= \lambda^2 \Var (H\mid Y)+ (1-\lambda)^2\Var (\tilde A\mid Y),
\]
which gives the expression for $\Var (\tilde A\mid Y)$.

Finally, we have
\begin{align*}
\Cov(H,A\mid Y)
&= \sum_{i\in A\cap H} \frac{1}{|\cH|}\frac{1}{|\cA|}\cdot \Var(s_i\mid Y) = \frac{|\cA\cap \cH|}{|\cH||\cA|}N\tau^{-1}.
\end{align*}
Similarly,
\[
\Var(H\mid Y)
= \sum_{i\in H}\frac{1}{|\cH|^2}\cdot \Var(s_i\mid Y)
= \frac{1}{|\cH|}\cdot N\tau^{-1}.
\]
Therefore,
\[
\lambda
= \frac{\Cov(H,A\mid Y)}{\Var(H\mid Y)}
= \frac{\frac{|\cA\cap \cH|}{|\cH||\cA|}}{\frac{1}{|\cH|}}
= \frac{|\cA\cap \cH|}{|\cA|}.
\]

\subsection{Proof of Proposition~\ref{prop:comparative_statics_bayes}}
The marginal information value $v(A\mid H)$ can be written as
\[
v(A\mid H)= \frac{1}{\tau_0+\tau_H}-\frac{1}{\tau_0+\tau_H+\tilde\tau}
\]
where
\[
\tilde\tau= \frac{(1-\lambda)^2}{\tau_A^{-1}-\lambda^2\tau_H^{-1}}
\]
First, we show that $\tilde\tau$ is increasing in $\tau_A$, deceasing in $\tau_H$, and decreasing in $\lambda$. The first two is straightforward; for the third, note that
\begin{align*}
\frac{\partial \tilde \tau}{\partial \lambda}&=\frac{2\lambda \tau_H^{-1}(1-\lambda)^2-2(1-\lambda)(\tau_A^{-1}-\lambda^2\tau_H^{-1})}{(\tau_A^{-1}-\lambda^2\tau_H^{-1})^2}\\
&=\frac{2(1-\lambda)(\lambda \tau_H^{-1}-\tau_A^{-1})}{(\tau_A^{-1}-\lambda^2\tau_H^{-1})^2}<0\\
\end{align*}
Therefore, $v(A\mid H)$ is increasing in $\tau_A$ and decreasing in $\lambda$ since $\tau_A$ and $\lambda$ affects $v(A\mid H)$ only through $\tilde\tau$. Finally, $v(A\mid H)$, as a function of $\tau_H$ and $\tilde\tau$, is decreasing in $\tau_H$ and increasing in $\tilde\tau$; combing the above result, $v(A\mid H)$ is decreasing in $\tau_H$.

\subsection{Proof of Proposition~\ref{prop:withhold_threshold}}
The difference between $L_{H-AI}$ and $L_H$ can be computed as follows: (recall $T=\tau_0+\tau_H+\tau_A$)
\begin{align*}
    L_{H-AI}-L_H&=\frac{1}{T}+\frac{2\lambda\tau_A}{T^2}-\frac{1}{\tau_0+\tau_H} \\
    &= \frac{(\tau_0+\tau_H)(\tau_0+\tau_H+(1+2\lambda)\tau_A)-(\tau_0+\tau_H+\tau_A)^2}{T^2(\tau_0+\tau_H)}\\
    &= \frac{\tau_A((2\lambda-1)(\tau_0+\tau_H)-\tau_A)}{T^2(\tau_0+\tau_H)}\\
\end{align*}
Therefore, if $\lambda<\frac12$, we always have $L_{H-AI}-L_H<0$.

If $\lambda>\frac12$, $L_{H-AI}-L_H\leq0$ if and only if
\[
(2\lambda-1)(\tau_0+\tau_H)- \tau_A\leq 0,
\]
which is equivalent to $\tau_A\geq \tau_{\rm aug}(\lambda)$ as desired.
\subsection{Proof of Proposition~\ref{prop:automation}}
Solving for $L_{AI}<L_{H-AI}$, we obtain that 
\[
L_{AI}<L_{H-AI} \Leftrightarrow 2\lambda \tau_A(\tau_0+\tau_A)-\tau_H(\tau_0+\tau_H+\tau_A)>0.
\]
The quadratic has a unique positive root \[
\tau_{\mathrm{auto}}(\lambda)= \frac{\tau_H-2\lambda\tau_0+\sqrt{(\tau_H-2\lambda\tau_0)^2+8\lambda\tau_H(\tau_0+\tau_H)}}{4\lambda},
\]
so that we further have
\[
L_{AI}<L_{H-AI} \Leftrightarrow\tau_A > \tau_{\mathrm{auto}}(\lambda)
\]

Moreover, it is straightforward that
\[
L_{AI}<L_{H} \Leftrightarrow \tau_A>\tau_H.
\]
Therefore,
\[
L_{AI}<\min\{L_{H-AI},L_H\}\Leftrightarrow \tau_A>\max\{\tau_{\mathrm{auto}}(\lambda),\tau_H\}.
\]
Finally, solving for $\tau_{\mathrm{auto}}(\lambda)\geq \tau_H$, we obtain
\[
\tau_{\mathrm{auto}}(\lambda)\geq \tau_H\Leftrightarrow \lambda \leq \bar \lambda = \frac{1}{2}+\frac{\tau_H}{2(\tau_0+\tau_H)}.
\]
The rest of the statements in the proposition follows directly.

\subsection{Proof of Proposition~\ref{prop:complementarity}}
From the above proof, we know that for complementarity $L_{\mathrm{H-AI}} < \min\{L_{\mathrm{AI}}, L_{\mathrm{H}}\}$ to hold, we need 
\[
\tau_A<\tau_{\mathrm{auto}}(\lambda)\text{ and } \tau_A>\tau_{\rm aug}(\lambda).
\]
Therefore, complementarity is possible if and only if the above condition is possible, which is equivalent to 
\[
\tau_{\rm aug}(\lambda)<\tau_{\mathrm{auto}}(\lambda).
\]
Solving for the range of $\lambda$, the above condition is again equivalent to 
\[
\lambda < \bar \lambda = \frac{1}{2}+\frac{\tau_H}{2(\tau_0+\tau_H)}
\]

If $\lambda \leq \frac12$, then $L_{H-AI}<L_H$ holds for any $\tau_A>0$, so that complementarity only needs $L_{H-AI}<L_{AI}$, which is equivalent to $\tau_A<\tau_{\mathrm{auto}}(\lambda)$.

If $\frac12<\lambda<\bar \lambda$, then complementarity is equivalent to 
\[
\tau_A<\tau_{\mathrm{auto}}(\lambda)\text{ and } \tau_A>\tau_{\rm aug}(\lambda).
\]
i.e., $\tau_A\in(\tau_{\rm aug}(\lambda), \tau_{\rm auto}(\lambda))$.
\end{appendices}

\end{document}